\documentclass[11pt,reqno]{article}
\usepackage{epsf, amsmath, amssymb, graphicx, epsfig, hyperref, mathtools,amsthm,bbm,xcolor,pifont}  
\usepackage[utf8]{inputenc}
\usepackage[margin=1in]{geometry}
\usepackage{verbatim}
\usepackage{caption,subcaption}

\usepackage[overload]{empheq}
\usepackage{cleveref}

\newtheorem{theorem}{Theorem}[section]
\newtheorem{proposition}{Proposition}[section]
\newtheorem{lemma}{Lemma}[section]

\newtheorem{assump}{Assumption}

\newtheorem{remark}{Remark}[section]
\numberwithin{equation}{section}

\renewcommand{\P}{\mathbb{P}}

\newcommand{\R}{\mathbb{R}}
\newcommand{\E}{\mathbb{E}}
\newcommand{\cE}{\mathcal{E}}
\newcommand{\N}{\mathbb{N}}
\newcommand{\F}{\mathcal{F}}

\newcommand{\B}{\mathcal{B}}
\newcommand{\cH}{\mathcal{H}}

\newcommand{\cM}{\mathcal{M}}

\newcommand{\T}{\mathcal{T}}

\newcommand{\cD}{\mathcal{D}}

\newcommand{\cZ}{\mathcal{Z}}
\newcommand{\tZ}{\widetilde{Z}}
\newcommand{\cW}{\mathcal{W}}

\newcommand{\cC}{\mathcal{C}}
\newcommand{\cG}{\mathcal{G}}
\newcommand{\M}{\mathcal{M}}
\newcommand{\cP}{\mathcal{P}}
\newcommand{\A}{\mathcal{A}}
\newcommand{\cL}{\mathcal{L}}
\newcommand{\cS}{\mathcal{S}}

\newcommand{\eps}{\varepsilon}

\newcommand{\ind}{\mathbbm{1}}

\DeclareMathOperator*{\esssup}{ess\,sup}
\DeclareMathOperator*{\essinf}{ess\,inf}

\newcommand{\nada}[1]{}
\definecolor{gb}{rgb}{0, 0.2, 0.8}

\title{Epstein-Zin Utility Maximization on a Random Horizon}
\author{Joshua Aurand\thanks{
		Verus Research, 6100 Uptown Blvd NE, Suite 260, Albuquerque, NM 87110, USA, email: \texttt{joshua.aurand@verusresearch.net}.} \and Yu-Jui Huang\thanks{
		University of Colorado, Department of Applied Mathematics, Boulder, CO 80309-0526, USA, email: \texttt{yujui.huang@colorado.edu}. Partially supported by National Science Foundation (DMS-1715439) and the University of Colorado (11003573).}}

\begin{document}
	\maketitle
	\begin{abstract}
		This paper solves the consumption-investment problem under 
		Epstein-Zin preferences on a random horizon. In an incomplete market, we take the random horizon to be a stopping time adapted to the market filtration, generated by all observable, but not necessarily tradable, state processes. Contrary to prior studies, we do not impose any fixed upper bound for the random horizon, allowing for truly unbounded ones. Focusing on the empirically relevant case where the risk aversion and the elasticity of intertemporal substitution are both larger than one, we characterize the optimal consumption and investment strategies using backward stochastic differential equations with superlinear growth on unbounded random horizons. This characterization, compared with the classical fixed-horizon result, involves an additional stochastic process that serves to capture the randomness of the horizon. As demonstrated in two concrete examples, changing from a fixed horizon to a random one drastically alters the optimal strategies. 
	\end{abstract}
	
\smallskip

\smallskip


	\section{Introduction}\label{Intro}
	Classical time-separable utilities unintentionally impose an artificial relation between an agent's risk aversion (denoted by $\gamma$) and elasticity of intertemporal substitution (EIS, denoted by $\psi$): the latter has to be the reciprocal of the former. Such a relation is widely rejected empirically. Bansal and Yaron \cite{Bansal04}, Bansal \cite{Bansal07}, Bhamra \cite{Bhamra10}, and Benzoni \cite{Benzoni11} all point to the fact $\psi>1$ from empirical data, 
	while estimations in Vissing-J{\o}rgensen and Attanasio \cite{Vissing03}, Bansal and Yaron \cite{Bansal04}, and Hansen et al. \cite{Hansen07} indicate $\gamma>1$. To disentangle EIS from risk aversion, Epstein and Zin \cite{EZ89} specifies a recursive utility in discrete time, whose continuous-time counterpart is formulated in Duffie and Epstein \cite{DuffieEpstein92}. These Epstein-Zin type utilities have proved instrumental in resolving observed market anomalies; see \cite{Bansal04, Bansal07, Bhamra10, Benzoni11}, among others.  
	
Since the seminal works \cite{EZ89, DuffieEpstein92}, the consumption-investment problem under Epstein-Zin preferences has been extensively studied on a fixed time horizon $T>0$, see e.g. \cite{DuffieLions, Schroder96, Kraft13, Seifried16, Kraft17, Xing}, while the infinite horizon case was recently approached in \cite{Melnyk20}. 
In practice, an agent need not have a fixed planning horizon in mind, either finite or infinite, upon entering the market. The time to exit can be random, depending on various factors in the market. 

In this paper, we study optimal consumption and investment under Epstein-Zin preferences on a random horizon. We focus on the empirically relevant case $\gamma, \psi>1$, and consider an incomplete market where an agent observes all the state processes, but cannot trade all of them. The random horizon $\tau$ is taken to be a stopping time adapted to the market filtration, generated by all observable (but not necessarily tradable) state processes. That is, the realization of $\tau$ depends on market conditions, but the involved uncertainty can be hedged against only {\it partially} through trading. 

Prior studies on a random horizon $\tau$, all under time-separable utilities, include \cite{Yaari65, Merton69, GH19} (where $\tau$ is independent of the market), \cite{Karatzas00, Kraft06} (where $\tau$ depends completely on the market), \cite{Blanchet08, Monique15} (where $\tau$ depends on the market and other external factors), among others. In all these works, $\tau$ is required {\it a priori} to be bounded (i.e. $\tau\le T$ a.s. for a known $T>0$). That is, a fixed known horizon $T>0$ is still in place, only less explicitly. Our framework, by contrast, dismisses the presence of any fixed horizon, allowing for truly unbounded random horizons. Moreover, all the above, except \cite{Monique15}, require market completeness, which we relax for practical applications. 

Our major finding is that a random horizon $\tau$ drastically alters the optimal consumption and investment strategies. Compared with the fixed-horizon result in Xing \cite{Xing}, our optimal strategies, in \eqref{optimalstrategies} below, involve an additional process $\hat Z$. While superfluous for the fixed-horizon case, $\hat Z$ serves to capture the effect of the random horizon $\tau$ on the optimal strategies; see Remarks~\ref{rem:hat Z} and \ref{rem:hat Z=0} for details.

Two concrete examples illustrate our finding explicitly. In Section~\ref{Example 1}, the random default time of a firm is shown to reduce the optimal consumption and investment from their classical levels in the fixed-horizon case, more significantly when the firm is closer to bankruptcy. In addition, we find that the optimal investment ratio is sensitive to, and actually increases with, EIS $\psi$---in contrast to prior findings on a fixed horizon; see Figure~\ref{fig:cons.invest.default} and the discussion below it. In Section~\ref{subsec:Heston}, a Heston model of stochastic volatility is used to exemplify the {\it fundamental risk} and the {\it noise trader risk} faced by a rational investor, motivated by the seminal work De Long et al.\ \cite{DeLong90}. 
We show that while the fixed-horizon optimal strategies dictate a constant proportion of wealth in the risky asset, the consideration of a random exit time---in the event that noise traders distort the asset price too significantly---changes the optimal investment ratio into a function of market states; see Figures~\ref{fig:zeroEpsilon}, \ref{fig:positiveEpsilon} and the discussion below them. 


Our analysis is based on techniques of backward stochastic differential equations (BSDEs). On a fixed horizon, BSDEs are fairly versatile for Epstain-Zin utility maximization, as shown in \cite{Xing}. A random horizon, nonetheless, poses a series of challenges.   

The first step of our investigation is to prove the existence of the Epstein-Zin utility process, given a consumption stream $(c_t)_{t\ge 0}$. This translates into solving a random-horizon BSDE with {\it non-uniform} superlinear growth: its generator grows super-linearly in one variable and the growth is not uniform in other variables. 
The literature of BSDEs on a random horizon $\tau$ frequently imposes ``$\tau\le T$ a.s. for a fixed $T>0$'', a condition we aim to relax. Among the few results that tackle unbounded $\tau$ (see e.g. \cite{Darling97, Pardoux99, Kobylanski00, Royer04, Briand00}), none of them allows for non-uniform superlinear growth. 
In response, we introduce a truncated BSDE on the interval $[0,n]$, for all $n\in\N$. With a fixed horizon $n$, the construction in \cite{Xing} can be used to deal with the non-uniform superlinear growth, whence a unique solution to each truncated BSDE exists. Motivated by Pardoux \cite{Pardoux99} and Briand and Carmona \cite{Briand00}, we prove that this sequence of solutions is Cauchy in a complete space of stochastic processes. The limit, as $n\to\infty$, exists and solves the original random-horizon BSDE. See Proposition~\ref{P1} and Theorem~\ref{thm:EZ exists} for details.   

Next, we look for consumption and investment strategies that maximize the Epstein-Zin utility. By dynamic programming, we derive a random-horizon BSDE, i.e. \eqref{EQ11} below, from which the candidate optimal strategies can be derived. This BSDE is non-standard: its generator has {\it quadratic} and {\it exponential} growth in several different variables. 
By a delicate truncation technique, we contain the exponential growth, making the generator grow linearly and remain strictly increasing in one of its variables---whence the results for random-horizon quadratic BSDEs can now be applied. Specifically, a careful use of the existence result in Briand and Confortola \cite[Theorem 3.3]{Confortola08}, followed by a comparison principle in Kobylanski \cite[Theorem 2.3]{Kobylanski00}, yields a solution to our non-standard BSDE; see Proposition~\ref{P5}. Note that the truncation technique we use is different from that in \cite{Xing}, as the latter requires a fixed horizon; see Remark~\ref{rem:new truncation}.

With the candidate optimal strategy $(\pi^*,c^*)$ derived, it remains to show its optimality among 
the set of {\it permissible} strategies, defined in \eqref{P} below. BMO arguments, useful in the fixed-horizon case to establish permissibility (as shown in \cite{Xing}), no longer work in the present setting: the BMO norms can easily blow up on an unbounded random horizon $\tau$. 
In view of this, we directly impose appropriate exponential moment conditions on $\tau$ (i.e. Assumption~\ref{A4} below), from which the permissibility of $(\pi^*,c^*)$ can be established; see Lemma~\ref{C4}. The optimality of $(\pi^*,c^*)$ then follows from standard arguments; see Theorem~\ref{T1}, the main result of this paper. 
Note that the conditions imposed on $\tau$ are less restrictive than what they might seem---they readily cover all prior studies on random-horizon consumption-investment problems (where ``$\tau\le T$ a.s. for a fixed $T>0$'' is imposed) and are commonly-seen in the literature of random-horizon BSDEs; see Remark~\ref{rem:not restrictive}.
 		

It is of interest to consider more general random horizons $\tau$ that are not limited to the market filtration. The companion paper \cite{AH20} pursued this direction: it studies optimal consumption, investment, and healthcare spending under Epstein-Zin preferences over an agent's random lifetime; namely, $\tau$ is the death time of the agent, which need not depend on the financial market. 

The rest of the paper is organized as follows. Section~\ref{EZ} establishes the existence and uniqueness of the Epstein-Zin utility process, for a given consumption stream. Section~\ref{CI} introduces the consumption-investment problem, derives the candidate optimal strategies, and proves that they are indeed optimal. By connecting the BSDE framework to the settings of hitting times, Section~\ref{Examples} illustrates numerically how a random horizon drastically alters the optimal strategies in two concrete financial models. Appendices contain auxiliary results and all the proofs.

	


	\section{Epstein-Zin Preferences on a Random Horizon}\label{EZ}
	Let $(\Omega,\mathcal{F},\mathbb{P})$ be a probability space that supports a $d$-dimensional Brownian motion $(B_{t})_{t\ge 0}$. Let $\mathbb{F}=(\mathcal{F}_{t})_{t\ge 0}$ be the $\P$-augmentation of the natural filtration generated by $B$ and $\T$ be the set of all $\mathbb F$-stopping times that are finite a.s. 
	
	Let us fix a random horizon $\tau\in \T$.  An agent obtains utility from a consumption stream $c=(c_t)_{0\le t\le  \tau}$, a nonnegative progressively measurable process, defined on the random horizon $[0,\tau]$. Here, $c_t$ represents the consumption rate at time $t$ for all $0\le t < \tau$, while $c_\tau$ stands for a lump-sum consumption at time $\tau$. Let $\delta > 0$ be the discount rate, $\gamma > 0 \neq 1$ be the relative risk aversion, and $\psi > 0$ be the elasticity of intertemporal substitution (EIS). Assume that the bequest utility function of the agent is $U(c) = \frac{c^{1-\gamma}}{1-\gamma}$. Then, given a consumption stream $c$, the \textit{Epstein-Zin utility process} on the random horizon $\tau$ is a process $V^c = (V_{t}^{c})_{t\ge 0}$ that satisfies
	\begin{equation}\label{EQ2}
	V_{t}^{c} 
	= \mathbb{E}_{t}\bigg[\int_{t\wedge\tau}^{\tau}f(c_s,V_{s}^{c})ds + \frac{c^{1-\gamma}_{\tau}}{1-\gamma}\bigg],\quad \forall t\ge 0,
	\end{equation} 
	where $\mathbb{E}_{t}[\cdot]$ denotes $\mathbb{E}[\cdot\mid\mathcal{F}_{t}]$ and the function $f(c,v)$, 
	the {\it Epstein-Zin aggregator}, is defined by 
	\begin{equation}\label{f}
	\begin{split}
	f(c,v)&:=\delta\frac{(1-\gamma)v}{1-\frac{1}{\psi}}\left(\bigg(\frac{c}{((1-\gamma)v)^{\frac{1}{1-\gamma}}}\bigg)^{1-\frac{1}{\psi}}-1\right)\\
	&=\delta\frac{c^{1-\frac{1}{\psi}}}{1-\frac{1}{\psi}}\big((1-\gamma)v\big)^{1-\frac{1}{\theta}}-\delta\theta v,\qquad\text{with}\quad \theta := \frac{1-\gamma}{1-\frac{1}{\psi}}.
	\end{split}
	\end{equation}
	In this paper, we focus on the specification $\gamma, \psi>1$, which is the empirically relevant case, as discussed in the introduction. Note that this implies $\theta<0$, which will be used frequently. 
	
	The goal of this section is to establish existence and uniqueness of the Epstein-Zin utility process $V^c$ in \eqref{EQ2}. This has been done for the fixed-horizon case (i.e. $\tau\equiv T$ for a fixed $T>0$) or the $\theta>0$ case; see e.g. \cite{Schroder96, Kraft17, Xing}.   
	We will construct $V^c$ in \eqref{EQ2} via the BSDE
	\begin{equation}\label{EQ3}
	V^c_{t} = \frac{c^{1-\gamma}_{\tau}}{1-\gamma} + \int_{t\wedge\tau}^{\tau}f(c_s,V^c_{s})ds - \int_{t\wedge\tau}^{\tau}Z^c_{s}dB_s ,\quad \forall t\ge 0.
	\end{equation}
	As observed in \cite{Xing}, with $\gamma,\psi>1$, $f(c,v)$ has superlinear growth in $v$ and is thus non-Lipschitz. Following the transformation in \cite[Section 2.1]{Xing}, we consider $(Y_{t},Z_{t}) := e^{-\delta\theta t}(1-\gamma)(V^c_t,Z^c_t)$, with the corresponding BSDE
	\begin{equation}\label{EQ5}
	Y_{t} = e^{-\delta\theta\tau}c^{1-\gamma}_{\tau} + \int_{t\wedge\tau}^{\tau}F(s,c_s,Y_{s})ds - \int_{t\wedge\tau}^{\tau}Z_{s}dB_s , \quad \forall t\ge 0,
	\end{equation}
	where  
	\[
	F(t,c,y) : = \delta\theta e^{-\delta t}c^{1-\frac{1}{\psi}}y^{1-\frac{1}{\theta}}.
	\]
	It is expected that \eqref{EQ5} is more manageable as it satisfies the {\it monotonicity condition}: $F(t,c,y)$ is decreasing in $y$, thanks to $\theta<0$ and $y\ge 0$. The set of admissible consumption streams is taken as   
	\begin{equation}\label{C}
	\mathcal{C}:=\bigg\{c\in\mathcal{R}_{+}:\mathbb{E}\bigg[\int_{0}^{\tau}e^{-2\delta s}c_{s}^{2(1-\frac{1}{\psi})}ds\bigg]<\infty\quad \text{and}\quad \mathbb{E}\bigg[e^{-2\delta\theta(2-\frac{1}{\theta})\tau}c_{\tau}^{2(2-\frac{1}{\theta})(1-\gamma)}\bigg]<\infty \bigg\},
	\end{equation}
	where $\mathcal{R}_{+}$ is the set of all nonnegative 
	progressively measurable processes. 
	
	\begin{remark}\label{rem1}
	Our admissible set $\mathcal{C}$ is larger than the commonly-used one under Epstein-Zin utilities on a fixed horizon (i.e. $\tau\equiv T$ for a fixed $T>0$). For example, \cite{Schroder96} requires 
	\[
	\mathbb{E}\bigg[\int_{0}^{T}c_{t}^{\ell}dt\bigg]<\infty\quad \text{and}\quad \mathbb{E}[c_{T}^{\ell}]<\infty,\quad \hbox{for all}\ \ell\in\mathbb{R}.
	\]
	Recently,  \cite{Xing} proposed a much weaker condition
	\[
	\mathbb{E}\bigg[\int_{0}^{T}e^{-\delta s}c_s^{1-{1}/{\psi}}ds\bigg] < \infty\quad \hbox{and}\quad  \mathbb{E}\big[c_T^{1-\gamma}\big]< \infty. 
	\]
    The integrability imposed in \eqref{C} is stronger than this, yet with a purpose---As will be seen, the additional integrability helps extend results in \cite{Xing} from a fixed horizon to a random one.	\end{remark}
	
	To state the main result of this section, let us introduce the following notation.  
	\begin{itemize}
		\item For any $q>1$, let $\mathcal{S}_{q}$ denote the set of $\R$-valued progressively measurable processes $Y$ such that $\|Y\|_{\mathcal{S}_{q}}^{q}:=\mathbb{E}[\sup_{t\ge 0}|Y_{t\wedge\tau}|^{q}]<\infty$. 
		\item Let $\mathcal S_\infty$ denote the set of $\R$-valued progressively measurable processes $Y$ such that $\|Y\|_\infty := \inf\{C\ge 0 : |Y_t|\le C\ \forall t\ge 0\ \hbox{a.s.}\} <\infty$.
		\item For any $q>1$, let $\mathcal{M}_{q}$ denote the set of $\R^d$-valued progressively measurable processes $Z$ such that $\|Z\|_{\mathcal{M}_{q}}^{q}:=\mathbb{E}[(\int_{0}^{\tau}\|Z_{t}\|^{2}dt)^{\frac{q}{2}}]<\infty$.
		\item For any $q>1$, $\mathcal{B}_{q} := \mathcal{S}_{q}\times\mathcal{M}_{q}$, with the norm $\|(Y,Z)\|_{\mathcal{B}_{q}}^{q} := \|Y\|_{\mathcal{S}_{q}}^{q}+\|Z\|_{\mathcal{M}_{q}}^{q}$. 
	\end{itemize}
	
	\begin{proposition}\label{P1}
		Suppose $\gamma,\psi>1$ and $c\in\mathcal{C}$. Then, \eqref{EQ5} admits a unique solution $(Y,Z)$ in $\mathcal{B}_{2}$ with $Y\ge 0$ a.s.
	\end{proposition}
	
	The proof of Proposition~\ref{P1} is relegated to Appendix~\ref{subsec:proof of Proposition P1}. 
	
	\begin{remark}
	Proposition~\ref{P1} does not follow from \cite[Theorem 5.2]{Briand03}, contrary to what it may seem. Indeed, our condition $c\in\mathcal{C}$ is weaker than the integrability required in \cite{{Briand03}}. For \cite[Theorem 5.2]{Briand03} to be applicable, we would need a stronger condition: there exists $\eps>0$ such that
    \begin{equation*}
   	\mathbb{E}\bigg[\int_{0}^{\tau}e^{-2(\delta-\eps) s}c_{s}^{2(1-\frac{1}{\psi})}ds\bigg]<\infty\quad\text{and}\quad \mathbb{E}\bigg[e^{-2(\delta+\eps)\theta(2-\frac{1}{\theta})\tau}c_{\tau}^{2(2-\frac{1}{\theta})(1-\gamma)}\bigg]<\infty.
    \end{equation*}
	\end{remark}
	
	\begin{remark}
	To prove Proposition~\ref{P1}, Appendix~\ref{subsec:proof of Proposition P1} devises a sequence of solutions, using the fixed-horizon construction in \cite{Xing}, whose limit ultimately solves \eqref{EQ5}. Alternatively, one could devise a sequence of solutions via bounded consumption streams, in line with \cite[Theorem 3.1]{Royer04}, whose limit would solve \eqref{EQ5} by a monotonicity argument. The proof, however, is no simpler than Appendix~\ref{subsec:proof of Proposition P1}, and the monotonicity argument also requires conditions like \eqref{C}. 
	\end{remark}
	
	The Epstein-Zin utility process can now be constructed.
	
	\begin{theorem}\label{thm:EZ exists}
		Suppose $\gamma,\psi>1$. For any $c\in\cC$, let $(Y,Z)$ be the unique solution of \eqref{EQ5} in $\B_2$. Then, $(V^c_t,Z^c_t) := \frac{e^{\delta \theta t}}{1-\gamma}(Y_t,Z_t)$ is the unique solution to \eqref{EQ3} in $\mathcal{B}_{2}$ that satisfies \eqref{EQ2} a.s.  
	\end{theorem}

The proof of Theorem~\ref{thm:EZ exists} is relegated to Appendix~\ref{subsec:proof of Theorem EZ exists}. 

\begin{remark}
It is worth noting that $Y, V^c\in \mathcal{S}_{2}$ particularly implies that they are of class D.
\end{remark}


\section{The Consumption-Investment Problem}\label{CI}
In this section, we introduce the consumption-investment problem in an incomplete market, under the framework of Section~\ref{EZ}. By dynamic programming, we derive a BSDE on a random horizon $\tau$, 
from which the candidate optimal strategies can be deduced. Under appropriate conditions on market coefficients (Assumption~\ref{A1}), a solution to the BSDE exists; see Proposition~\ref{P5}. On strength of exponential moment conditions on $\tau$ (Assumption~\ref{A4}), the candidate strategies, given in \eqref{optimalstrategies} below, are indeed optimal among an appropriate class of strategies; see Theorem~\ref{T1}.  


\subsection{The Markovian Setup}\label{subsec:setup}
We take up the framework in Section~\ref{EZ}, with $B=(W,\hat W)$ a two-dimensional Brownian motion (i.e. $d=2$). Let $E$ be an open domain in $\R$, and consider an $E$-valued state process 
\begin{equation}\label{EQ8}
dY_{t} = a(t,Y_{t})dt + b(t,Y_{t})dW_{t},\quad Y_{0} = y\in E,
\end{equation}
where $a,b:\mathbb{R}_{+}\times E\rightarrow\mathbb{R}$ are given Borel measurable functions. 
The market consists of a riskfree asset $S^{0}$ and a risky asset $S_{t}$, satisfying the dynamics
\begin{equation}\label{B&S}
\begin{split}
dS_{t}^{0}&=r(t,Y_{t})S_{t}^{0}dt,\\
dS_{t}&=S_{t}\left((r(t,Y_{t}) +\lambda(t,Y_{t}))dt+\sigma(t,Y_{t})\left(\rho(t,Y_{t}) dW_{t} + \hat{\rho}(t,Y_{t})d\hat{W}_{t}\right)\right),
\end{split}
\end{equation}
where $r, \lambda, \sigma,\rho,\hat{\rho}:\mathbb{R}_{+}\times E\rightarrow\mathbb{R}$ are given Borel measurable functions. In particular, $\rho$ and $\hat{\rho}$, called the correlation functions, satisfy $\rho^{2}(t,y)+\hat{\rho}^2(t,y) = 1$ for all $(t,y)\in \mathbb{R}_{+}\times E$.  

An agent, with initial wealth $x>0$, must decide a proportion $\pi_{t}\in \R$ of wealth to invest in the risky asset and a consumption rate ${c}_{t}\ge 0$ at every moment $t\ge 0$ before the random terminal time $\tau$. The corresponding wealth process $X^{\pi,c}$ is given by 
\begin{align}\label{EQ13}
dX_{t}&= X_{t}\left[(r_{t}+\pi_{t}\lambda_{t}))dt + \pi_{t}\sigma_{t}\big(\rho_{t} dW_{t} + \hat{\rho}_{t}d\hat{W}_{t}\big)\right]-{c}_{t} dt, \notag\\
&= X_{t}\left[(r_{t}+\pi_{t}\lambda_{t}))dt + \pi_{t}\sigma_{t} dW^\rho_{t}\right]-{c}_{t} dt,\quad X_{0}= x,
\end{align}
where $r_t$, $\lambda_t$, $\sigma_t$, $\rho_{t}$, $\hat{\rho}_{t}$ represent $r(t,Y_t)$, $\lambda(t,Y_{t})$, $\sigma(t,Y_{t})$, $\rho(t,Y_{t})$, $\hat{\rho}(t,Y_{t})$, respectively, and 
\[
	W_{t}^{\rho} := \int_{0}^{t}\rho_{s} dW_{s} + \int_{0}^{t}\hat{\rho}_{s}d\hat{W}_{s},\quad t\ge 0,
\]
is again a Brownian motion. 
We enforce the following conditions on the market coefficients.

\begin{assump}\label{A1}
The coefficients $\sigma, r, \lambda, \rho, \hat{\rho}, a,$ and $b$ are locally Lipschitz in $E$; the process $Y$ does not reach the boundary of $E$ in finite time a.s.; $\{r_{t\wedge \tau}\}_{t\ge 0}$ and $\{\frac{\lambda_{t\wedge \tau}}{\sigma_{t\wedge\tau}}\}_{t\ge 0}$ are bounded processes (i.e. belong to $\cS_\infty$); $\inf_K \sigma(t,y)>0$ and $\inf_K b(t,y)>0$ for any compact subset $K$ of $\R_+\times E$. 
\end{assump}

A strategy $(\pi,c)$ is called \textit{admissible} if it belongs to 
\[
\mathcal{A} : = \{(\pi,c) : c\in\mathcal{C},\ c_\tau = X^{\pi,c}_\tau,\ X_{t}^{\pi,c}> 0\ \hbox{for all $0\le t\le \tau$ a.s.}\}.
\] 
The agent intends to maximize her Epstein-Zin utility $V^c_0$ by choosing a pair $(\pi^*,c^*)$ from some appropriate collection $\mathcal P\subseteq\A$. That is, the goal is to attain the optimal value 
\begin{equation}\label{the problem}
V^*_0 := \sup_{(\pi,c)\in\mathcal{P}} V_{0}^{c},
\end{equation}
where $V^c$ is the solution to \eqref{EQ2} with $c_\tau = X^{\pi,c}_\tau$, by some strategy $(\pi^*,c^*)\in \mathcal P$. The collection $\mathcal P\subseteq \A$ is up to the agent's choice. In this paper, we will take $\mathcal{P}$ to be the set of {\it permissible} strategies, defined precisely in \eqref{P} below.



\subsection{The Ansatz}\label{subsec:ansatz}
Motivated by the classical decomposition of time-separable power utilities (see e.g. \cite[Section 3]{Pham02}) and the decomposition of the Epstein-Zin utility in \cite[(2.9)]{Xing} on a fixed horizon, 
we suspect that the optimal utility process $V^*$ 
can be decomposed into
\begin{equation}\label{EQ9}
V^*_{t} 
= \frac{X_{t\wedge\tau}^{1-\gamma}}{1-\gamma}e^{D_{t\wedge\tau}}\quad t\ge 0,
\end{equation}
where $D$ is a process satisfying the BSDE
\begin{equation}\label{EQ11}
D_{t} = \int_{t\wedge\tau}^{\tau}H(s,D_{s},Z_{s},\hat{Z}_{s})ds -\int_{t\wedge\tau}^{\tau}Z_{s}dW_{s}-\int_{t\wedge\tau}^{\tau}\hat{Z}_{s}d\hat{W}_{s},\quad t\ge 0,
\end{equation}
for some generator $H$ to be determined. Note that \eqref{EQ9} and \eqref{EQ2} 
suggest that the process
\begin{equation}\label{R}
t\mapsto \frac{X_{t\wedge\tau}^{1-\gamma}}{1-\gamma}e^{D_{t\wedge\tau}} + \int_{0}^{t\wedge\tau}f\left(c_{s},\frac{X_{s}^{1-\gamma}}{1-\gamma}e^{D_{s}}\right)ds
\end{equation}
should be a supermartingale for any $(\pi,c)\in\mathcal{P}$ and a martingale for an optimal strategy $(\pi^*,c^*)$. Detailed calculations, similar to those in \cite[p. 234]{Xing}, yield the drift term of the above process:
\begin{align*}
\frac{X_{t}^{1-\gamma}}{1-\gamma}e^{D_{t}}\bigg(&-H(t,D_{t},Z_{t},\hat{Z}_{t}) + \frac{Z_{t}^{2}+\hat{Z}_{t}^{2}}{2}+(1-\gamma)(r_{t}-\tilde{c}_{t}+\pi_{t}(\lambda_{t}+\sigma_{t}\rho_{t} Z_{t}+\sigma_{t}\hat{\rho}_{t} \hat{Z}_{t}))\\&\quad -\frac{\gamma(1-\gamma)}{2}(\pi_{t}\sigma_{t})^{2} +\delta\theta\tilde{c}_{t}^{1-\frac{1}{\psi}}e^{-\frac{D_{t}}{\theta}} - \delta\theta\bigg),
\end{align*}
where $\tilde{c}_t := c_t/X_t$ is the proportion of wealth consumed per unit of time. This indicates that
\begin{align}\label{EQ10}
H(t,D_{t},Z_{t},\hat{Z}_{t}) &= (1-\gamma)r_{t}+\frac{Z_{t}^{2}+\hat{Z}_{t}^{2}}{2}- \delta\theta + \inf_{\tilde c\ge 0} \left(-(1-\gamma)\tilde{c}+\delta\theta\tilde{c}^{1-\frac{1}{\psi}}e^{-\frac{D_{t}}{\theta}}  \right)\nonumber\\ 
&\hspace{0.2in}+\inf_{\pi\in \R} \left((1-\gamma)\pi\left(\lambda_{t}+\sigma_{t}\rho_{t} Z_{t}+\sigma_{t}\hat{\rho}_{t} \hat{Z}_{t}\right)-\frac{\gamma(1-\gamma)}{2} \pi^2\sigma^2_{t}\right).
\end{align}
Solving the involved minimization problems yields the candidate optimal strategies $(\pi^{*},\tilde{c}^{*})$:
\begin{equation}\label{optimalstrategies}
\pi^{*}_t  = \frac{\lambda_{t}+\sigma_{t}(\rho_{t} Z_{t}+\hat{\rho}_{t} \hat{Z}_{t})}{\gamma\sigma_{t}^{2}}\quad\hbox{and}\quad \frac{c^*_t}{X^*_t}=\tilde{c}^{*}_t = \delta^{\psi}e^{-\frac{\psi}{\theta}D_{t}}\qquad \forall t\in[0,\tau),
\end{equation}
where $X^* := X^{\pi^*,c^*}$ is the candidate optimal wealth process. Plugging these into \eqref{EQ10}, we have
\begin{equation*}
H(t,D_{t},Z_{t},\hat{Z}_{t}) = \frac{Z_{t}^{2}+\hat{Z}_{t}^{2}}{2} + (1-\gamma)r_{t} - \delta\theta + \frac{\delta^{\psi}\theta}{\psi}e^{-\frac{\psi}{\theta}D_{t}} + \frac{(1-\gamma)(\lambda_{t}+\sigma_{t}\rho_{t} Z_{t}+\sigma_{t}\hat{\rho}_{t} \hat{Z}_{t})^{2}}{2\gamma\sigma_{t}^{2}}.
\end{equation*}
Rearranging and simplifying terms  gives
\begin{align}\label{H}
H(t,D_{t},&\ Z_{t},\hat{Z}_{t})=\nonumber\\
&\frac{Z_{t}^{2}}{2}\left(1 + \frac{(1-\gamma)}{\gamma}\rho_{t}^{2}\right)+\frac{\hat{Z}_{t}^{2}}{2}\left(1 + \frac{(1-\gamma)}{\gamma}\hat{\rho}_{t}^{2}\right)+\frac{(1-\gamma)\lambda_{t}}{\gamma\sigma_{t}}\rho_{t}Z_{t}+\frac{(1-\gamma)\lambda_{t}}{\gamma\sigma_{t}}\hat{\rho}_{t}\hat{Z}_{t}\nonumber\\
&+ \frac{(1-\gamma)}{\gamma}\rho_{t}\hat{\rho}_{t}Z_{t}\hat{Z}_{t}+\frac{\delta^{\psi}\theta}{\psi}e^{-\frac{\psi}{\theta}D_{t}} + (1-\gamma)\bigg(r_{t} + \frac{\lambda_{t}^{2}}{2\gamma\sigma_{t}^{2}}\bigg)-\delta\theta.
\end{align}

\begin{remark}\label{rem:hat Z}
A significant departure from the fixed-horizon case is the involvement of $\hat Z$ in \eqref{optimalstrategies}-\eqref{H}. Indeed, the formulas in \cite[p.235]{Xing} can be obtained by taking $\hat Z\equiv 0$ in \eqref{optimalstrategies}-\eqref{H}, leading to a simpler generator and a more straightforward investment strategy. This simplification does not hold in our case: the randomness of $\tau$ can be fully captured only with the additional process $\hat Z$, as explained in detail in Appendix~\ref{sec:Markov}; see Remark~\ref{rem:hat Z=0} particularly. 
\end{remark}

To make sense of the heuristic derivations above, the first task is to show the existence of a solution to the BSDE \eqref{EQ11}, with the generator $H$ given in \eqref{H}. This can be tricky: $H$ has quadratic growth in both $Z$ and $\hat Z$, along with exponential growth in $D$ (as $\theta<0$). 

BSDEs with quadratic growth were investigated in \cite{Kobylanski00}, which had been expanded upon by \cite{Briand06, Briand07, Confortola08}. The results in \cite{Briand06, Briand07} seem promising to our case, yet the exponential growth in $D$ prohibits us from using them. While there is a truncation technique in \cite{Xing} to tame the exponential growth in $D$, it requires a bounded time horizon. Ultimately, we will devise a new truncation technique to curb the exponential growth in $D$ on a possibly unbounded random horizon; see Remark~\ref{rem:new truncation} for details. With exponential growth contained, a delicate use of \cite{Confortola08} and \cite{Kobylanski00} in sequence (see Remark \ref{rem:in order} for details) yields the following.
 
\begin{proposition}\label{P5}
	Suppose $\gamma,\psi>1$ and Assumption \ref{A1} holds. Then, 
	there exists a solution $(D,Z,\hat Z)\in\cS_{\infty}\times\cM_{2}$ to \eqref{EQ11}, with $H$ given in \eqref{H}.
\end{proposition}

The proof of Proposition~\ref{P5} is relegated to Appendix~\ref{subsec:proof of Proposition P5}. 

\begin{remark}\label{rem:another ansatz}
Besides \eqref{EQ9}, another useful decomposition is $V^*_{t} = \frac{X_{t}^{1-\gamma}}{1-\gamma}P_{t}^{k}$, for some process $P$ and $k\in\R$; see \cite{Zariphopoulou01, Pham02}. This ansatz can potentially generate a simpler BSDE, without quadratic or exponential growth as seen in \eqref{H}; see Section~\ref{Example 1} below. This simplification, however, only works when $\hat Z\equiv 0$ and the correlation function is constant, i.e. $\rho(t,y)\equiv \rho$. 
\end{remark}

\begin{remark}
On a fixed horizon, \cite[Proposition 2.9]{Xing}, analogous to Proposition~\ref{P5}, is established without boundedness of market price of risk $\lambda/\sigma$. Specifically, by a change of measure using Girsanov's theorem, the generator in \cite[(2.13)]{Xing}, analogous to $H$ in \eqref{H}, is simplified---so as to mitigate the effect of $\lambda/\sigma$. Girsanov's theorem however requires a fixed horizon $T>0$. As our random horizon $\tau$ can be unbounded (i.e. $\P(\tau> T)>0$ for all $T>0$), it is unclear how the same technique in \cite{Xing} can be applied here. 
Hence, we still impose boundedness of $\lambda/\sigma$ in Assumption \ref{A1}. Note that this is not an uncommon assumption even for the fixed-horizon case; see \cite{Kraft17, Schroder96}.
\end{remark}


\subsection{Verification}\label{subsec:verification}
With $(\pi^*,c^*)$ in \eqref{optimalstrategies} well-defined, thanks to Proposition~\ref{P5}, it remains to show its optimality among an appropriate set of strategies. A strategy $(\pi,c)$ is called {\it permissible} if it belongs to 
\begin{equation}\label{P}
\mathcal{P} := \{(\pi,c)\in \A: (X^{\pi,c}_\cdot)^{1-\gamma}\ \text{is of class $D$}\}. 
\end{equation}
Note that this is in line with the set of permissible strategies in \cite{Xing}. In the paragraph under \cite[Proposition 2.9]{Xing}, it is required that $(X^{\pi,c}_\cdot)^{1-\gamma}e^{D_\cdot}$ is of class $D$. This is equivalent to \eqref{P} as the process $D_\cdot$ is bounded in our setting (by Proposition~\ref{P5}). The aim of this subsection is to establish the optimality within $\mathcal{P}$ of the candidate $(\pi^{*},c^{*})$, using verification arguments.
 
Before we start, let us first motivate our definition of $\cP$ in \eqref{P}. 
 
\begin{remark}\label{rem:why class D}
A crucial step in verification is to show $w\ge V^*_0$, for some candidate optimal value function $w$. This is a known challenge when $\gamma>1$, even for time-separable utilities. By the standard verification arguments (see e.g., the proof of \cite[Theorem 3.5.2]{Pham-book-09}), proving $w\ge V^*_0$ requires 
\begin{equation}\label{Fatou}
\lim_{t\to\infty} \E[w(t\wedge\tau, X^{\pi,c}_{t\wedge \tau})]\ge \E[w(\tau, X^{\pi,c}_{\tau})]. 
\end{equation}
Under the normal ansatz $w(t,x) = \phi(t) \frac{x^{1-\gamma}}{1-\gamma}$ for a bounded $\phi:\R_+\to\R_+$, when $0<\gamma<1$, we have $w(\cdot\wedge\tau,X^{\pi,c}_{\cdot\wedge\tau})\ge 0$ and thus \eqref{Fatou} holds by Fatou's lemma. However, for $\gamma>1$, 
$w(t\wedge\tau,X^{\pi,c}_{t\wedge\tau})$ may approach $-\infty$ when $X^{\pi,c}_{t\wedge\tau}$ gets near $0$, leaving \eqref{Fatou} in question. 

Hence, for $\gamma>1$, we in general need ``$w(\cdot\wedge\tau,X^{\pi,c}_{\cdot\wedge\tau})$ is bounded from below by a uniformly integrable process'', so that Fatou's lemma can be applied to conclude \eqref{Fatou}. As $w(\cdot\wedge\tau,X^{\pi,c}_{\cdot\wedge\tau})\le 0$ by definition, the above condition amounts to ``$w(\cdot\wedge\tau,X^{\pi,c}_{\cdot \wedge\tau})$ is uniformly integrable''. This corresponds precisely to ``$w(\cdot\wedge\tau, X^{\pi,c}_{\cdot\wedge\tau})$ is of class $D$'' in \cite{Cheridito11}. Under Epstein-Zin preferences, this condition is followed by \cite{Xing} and \eqref{P} above.
\end{remark}

\begin{remark}
Simple examples of permissible strategies include $\eps$-modifications considered in \cite{Janecek12} and \cite{GH19}. If an agent is given additional wealth $\eps>0$ but leaves it untouched, any original strategy $(\pi,c)$ turns into a new one $(\pi^\eps,c^\eps)$ such that $X^{\pi^\eps,c^\eps}_{\cdot\wedge\tau} \ge \eps$; see the explicit construction in \cite[p. 355]{GH19}. As $ 0 < (X^{\pi^\eps,c^\eps}_{\cdot\wedge\tau})^{1-\gamma} \le \eps^{1-\gamma}$, we have $(\pi^\eps,c^\eps)\in\cP$.
\end{remark}

In showing that $(\pi^{*},c^{*})$ is permissible, the random horizon $\tau$ poses nontrivial challenges. As opposed to the fixed-horizon case, the boundedness of $D$ (by Proposition~\ref{P5}) does not imply that $\int_0^\cdot Z_t dW_t$ and $\int_0^\cdot \hat Z_t d\hat W_t$ are BMO-martingales. Indeed, in estimating the BMO norms, some constants involved are horizon-dependent---they can easily blow up on a potentially unbounded horizon $\tau$; see e.g. \cite[Lemma 3.1]{Morlais09} and the comment below \cite[(2)]{Hu05}. In other words, the BMO arguments, useful in establishing permissibility of strategies on a fixed horizon (see e.g. \cite[Lemma B.2]{Xing} and \cite[Lemma 3.1]{Morlais09}), 
do not apply in our random-horizon case. 

To proceed, we directly impose appropriate integrability conditions on the random horizon $\tau$, from which the permissibility of $(\pi^*,c^*)$ can be derived. To this end, hereon we set 
\begin{align}\label{Market-Constants}
C_{{\lambda}/{\sigma}} := \left\|\frac{\lambda_{t\wedge\tau}}{\sigma_{t\wedge\tau}}\right\|^2_\infty,\quad \overline{r} := \esssup\left(\sup_{t\ge 0}r_{t\wedge\tau}\right),\quad \underline{r} := \essinf\left(\inf_{t\ge 0} r_{t\wedge\tau}\right).
\end{align}
Note that these constants are finite because of Assumption~\ref{A1}. Also, we consider 
\begin{equation}\label{p's}
p_{+} := 2\left(1-\frac{1}{\psi}\right)>0\quad \hbox{and}\quad p_{-}:= 2\left(2-\frac{1}{\theta}\right)(1-\gamma)<0. 
\end{equation}

\begin{assump}\label{A4}
Let $(D,Z,\hat{Z})\in \cS_\infty\times \cM_2$ be the solution to \eqref{EQ11} in Proposition~\ref{P5}, and set $\widetilde Z_{t} := \rho_{t}Z_{t}+\hat{\rho}_{t}\hat{Z}_{t}$. We assume that there exists $q>1$ such that 
	\begin{equation*}
	\E\left[\exp\left(q\left(2\overline{r}p_{+}+\left(p_{+}+\frac{4p_{+}^2}{\gamma^2}\right)C_{{\lambda}/{\sigma}}\right)\tau+\frac{4qp_{+}^2}{\gamma^2}\int_{0}^{\tau}\widetilde Z_{s}^2 ds\right)\right]<\infty
	\end{equation*}
	and 
	\begin{equation*}
	\E\left[\exp\left(\left(\frac{-p_{-}\delta}{1-1/\psi}+2\underline{r}p_{-}-2p_{-}\delta^{\psi}e^{-\frac{\psi}{\theta}\widetilde{C}}+\left(|p_{-}|+\frac{4p_{-}^2}{\gamma^2}\right)C_{{\lambda}/{\sigma}}\right)\tau+\frac{4 p_{-}^2-2p_-}{\gamma^2}\int_{0}^{\tau}\widetilde Z_{s}^2 ds\right)\right]<\infty
	\end{equation*}
	where $\widetilde{C}:= \esssup\big(\sup_{t\ge 0} D_t\big) <\infty$. 
\end{assump}

\begin{remark}\label{rem:not restrictive}
Assumption~\ref{A4} is not restrictive in view of the literature. Prior studies of the consumption-investment problem on a random horizon $\tau$ (with time-separable utilities) all require $\tau\le T$ a.s. for a fixed $T>0$; see e.g. \cite{Blanchet08, Bouchard04, ElKaroui10, Kharroubi13, Monique15}. Assumption~\ref{A4} covers this case trivially, while allowing for suitable unbounded $\tau$. Moreover, this type of exponential moment conditions is common for random-horizon BSDEs, such as 
\cite[(A4)]{Briand00}, \cite[(c), Section 4]{Pardoux99}, and \cite[(25)]{Darling97}. 
\end{remark}

\begin{remark}
Assumption \ref{A4} can be relaxed to some extent, depending on the specific market model employed. For instance, in the Heston model considered in Section \ref{subsec:Heston} below, $\widetilde Z$ is a bounded process, which largely simplifies the exponential moment conditions. 
\end{remark}

With the aid of Assumption~\ref{A4}, we are able to derive the permissibility of $(\pi^{*},c^{*})$ in \eqref{optimalstrategies}. 

\begin{lemma}\label{C4}
Suppose $\gamma,\psi>1$ and Assumptions \ref{A1} and \ref{A4} hold. Then, $(\pi^{*},c^{*})$ defined in \eqref{optimalstrategies} belongs to $\mathcal P$.
\end{lemma}
The proof of Lemma~\ref{C4} is relegated to Section~\ref{subsec:proof for Lemma C4}.
With $(\pi^{*},c^{*})\in\mathcal{P}$, it remains to show that $(\pi^{*},c^{*})$ is optimal within $\mathcal P$; namely, it solves \eqref{the problem}. In view of the arguments in Section~\ref{subsec:ansatz}, this boils down to showing that the process in \eqref{R} is a supermartingale for each $(\pi,c)\in\mathcal{P}$ and a martingale for $(\pi^{*},c^{*})$.  This can be done by modifying the arguments in \cite[Theorem 2.14]{Xing}.

\begin{theorem}\label{T1}
Suppose $\gamma,\psi>1$ and Assumptions \ref{A1} and \ref{A4} hold. Then, $(\pi^{*},c^{*})$ defined in \eqref{optimalstrategies}, with $(D,Z,\hat Z)\in \mathcal S_\infty\times\cM_2$ a solution to \eqref{EQ11}, is a maximizer of \eqref{the problem}. Moreover, for any initial wealth $x>0$, the optimal Epstein-Zin utility is given by  $V_{0}^{*} = \frac{x^{1-\gamma}}{1-\gamma}e^{D_{0}}$. 
\end{theorem}
The proof of Theorem~\ref{T1} is relegated to Section~\ref{subsec:proof of Theorem T1}.

\begin{remark}
The permissibility of $(\pi,c)$ prevents the wealth process $X^{\pi,c}$ from getting too close to 0 for too long. Indeed, if $X^{\pi,c}$ is very close to 0 for a long time, $(X^{\pi,c})^{1-\gamma}$ will approach $+\infty$ very often and likely violate the permissibility of $(\pi,c)$. 
Economically, this means that even if one can find some $(\pi,c)\in\A$ that outperforms $(\pi^*,c^*)\in\cP$, it is likely that $(\pi,c)$ is {\it not} practically viable, as the induced wealth can be very small for an extended time period. 
\end{remark}


\section{Examples}\label{Examples}
In this section, we consider two concrete examples in Sections~\ref{Example 1} and \ref{subsec:Heston}, respectively. The purpose is twofold. First, we demonstrate explicitly that, compared with the classical fixed-horizon case, the randomness of the horizon drastically alters the optimal strategies; see Figures~\ref{fig:cons.invest.default}, \ref{fig:zeroEpsilon}, and \ref{fig:positiveEpsilon}. Moreover, the examples showcase, on the technical side, how Assumptions~\ref{A1} and \ref{A4} can be verified in practice. For the case where the random horizon $\tau$ is bounded (Section~\ref{Example 1}), one can check the assumptions in a straightforward way. For the case where $\tau$ is unbounded (Section~\ref{subsec:Heston}), a detailed estimation of the moment generating function of $\tau$ is needed; see Lemma~\ref{lem:MGF}. 


\subsection{The Default Problem: Bounded Time Horizon}\label{Example 1}
In this example, we take up the framework of Kraft and Steffensen \cite{Kraft06} and consider a consumption-investment problem where a defaultable bond is involved. This will generalize \cite{Kraft06, Korn03} from time-separable utilities to Epstein-Zin preferences. 
Specifically, we take 
\eqref{B&S} to be 
\begin{equation}\label{B&S''}
\begin{split}
	dS_{t}^{0} &= rS_{t}^{0}dt,\quad S_{0}^{0} = 1,\\
	dS_{t} &= S_{t}\left(\alpha dt + \sigma dW_{t}\right),\quad S_{0} = s>0,
\end{split}
\end{equation}
where $r,\alpha\in\R$ and $\sigma>0$ are given constants. Note that this is a complete market: there is no dependence on either the additional market state $Y$ in \eqref{EQ8} or the additional Brownian motion $\hat W$ in \eqref{B&S} (particularly, we take $
\rho\equiv1$ and $\hat{\rho}\equiv0$ in \eqref{B&S}). In addition, as all coefficients in \eqref{B&S''} are constant, Assumption \ref{A1} is readily satisfied.

In line with \cite{Kraft06, BC76}, the process $S$ represents the firm value of a company, which has issued stock and a zero-coupon bond (with maturity $T>0$ and face value $K>0$). The bond has a safety covenant: before its maturity $T>0$, if the firm value falls to a certain threshold (the {\it bankruptcy level}), the company is forced to bankruptcy and passed over to the bondholder. Following \cite{Kraft06}, we take the bankruptcy level to be
\[
\mathfrak{B}_t = L e^{-\kappa (T-t)}\quad \hbox{$t\in[0,T]$},
\] 
for some constants $L,\kappa\ge0$. 
If the bankruptcy level is never reached by time $T$, the company simply redeems the bond (i.e. pays the bondholder $K>0$) at maturity. 

An agent, who invests in the stock of the company, aims to achieve the optimal Epstein-Zin utility value, i.e. $V^*_0$ in \eqref{the problem}, over the random horizon 
\[
\tau := \overline{\tau}\wedge T\quad \text{with}\quad \overline{\tau}:=\inf\{t\ge0: S_{t} \le \mathfrak{B}_t\};
\]
namely, the time to the forced bankruptcy or the maturity of the bond, whichever comes first. As $\tau$ is bounded, the arguments in \cite[Lemma B.2]{Xing} and \cite[Lemma 3.1]{Morlais09} can be applied, showing that $\int_{0}^{\cdot}Z_{s}dW_{s}$ is a BMO-martingale. This property, along with the boundedness of $\tau$, ensures that Assumption \ref{A4} is satisfied; see also Remark~\ref{rem:not restrictive}. 

In the present complete market setting, the ansatz $V^*_{t} = \frac{X_{t}^{1-\gamma}}{1-\gamma}P_{t}^{\gamma}$, for some process $P$, will help reduce the nonlinearity in \eqref{EQ11}; see Remark~\ref{rem:another ansatz} for related discussions. Specifically, as $\hat W$ is not involved in \eqref{B&S''}, we can drop the last term in \eqref{EQ11} by taking $\hat Z\equiv 0$. Under the change of variables $P_{t} = e^{\frac{1}{\gamma}D_{t}}$ and $Q_t = \frac{1}{\gamma}Z_{t}e^{\frac{1}{\gamma}D_{t}}$, the reduced BSDE \eqref{EQ11} for $(D,Z)$ (with $\hat Z\equiv 0$, $\rho\equiv1$, and $\hat\rho\equiv 0$) turns into the BSDE
\begin{equation}\label{BSDE'''}
P_t = 1 + \int_{t\wedge\tau}^{\tau}\cH(P_s,Q_s)ds - \int_{t\wedge\tau}^{\tau}Q_s dW_{s},\quad t\ge 0,
\end{equation}
for $(P,Q)$, 
where
\begin{align*}
	\cH(p,q) :=\frac{(1-\gamma)(\alpha-r)}{\gamma\sigma} q + \theta\frac{\delta^{\psi}}{\gamma\psi}\left(p^{1-\frac{\gamma\psi}{\theta}}\right) + \frac{1-\gamma}{\gamma}\left(r + \frac{(\alpha-r)^{2}}{2\gamma\sigma^{2}}\right)p -\frac{\delta\theta}{\gamma}p.
\end{align*}
The optimal strategies \eqref{optimalstrategies} can then be written in terms of $(P,Q)$ as 
\begin{equation}\label{optimalstrategies'''}
\pi^*_t = \frac{(\alpha-r) P_t + \gamma \sigma  Q_t}{P_t\gamma\sigma^2} = \frac{\alpha-r}{\gamma\sigma^2} +\frac{Q_t}{\sigma P_t}\quad \text{and}\quad \tilde{c}^{*}_t = \displaystyle\delta^{\psi}P_t^{-\frac{\psi}{\theta}\gamma}.
\end{equation}
To solve \eqref{BSDE'''}, let us introduce the auxiliary state process
$\mathcal Y_{t} = -\ln S_{t}$, $t\ge 0$.
By taking 
\begin{equation}\label{P,Q,u}
P_t= u(t,\mathcal Y_t)\quad \hbox{and}\quad Q_t=- \sigma u_{y}(t,\mathcal Y_{t})
\end{equation}
for some $u:[0,T]\times \R$ to be determined, \eqref{BSDE'''} yields the PDE
\begin{equation}\label{PDE_Ex1}
	\begin{split}
		u_{t}(t,y)+\cH(u,-\sigma u_{y}) = \Big(\alpha-\frac{1}{2}\sigma^{2}\Big)u_{y}(t,y) - \frac{1}{2}\sigma^{2}u_{yy}(t,y),\quad &y<-\ln(L)+\kappa(T-t),\ t<T; \\
		u(t,y)=1,\quad &\hbox{otherwise}.
	\end{split}
\end{equation}
The introduction of $\mathcal Y$, superfluous at first glance, ensures that the coefficients in front of $u_y$ and $u_{yy}$ in \eqref{PDE_Ex1} are constants, which facilitates the derivation of Proposition~\ref{Prop: Parabolic} below. If $\mathcal Y$ in \eqref{P,Q,u} is replaced by the original state $S$, the resulting PDE will be much more complicated.  

\begin{proposition}\label{Prop: Parabolic}
	For any $\alpha, r\in \R$, $\sigma>0$, and $L,\kappa\ge0$, 
	there exists a unique bounded solution $u\in C^{1,2}([0,T]\times\R)$ to \eqref{PDE_Ex1}.
\end{proposition}

The proof of Proposition~\ref{Prop: Parabolic} is relegated to Appendix~\ref{subsec: Parabolic PDE}.
In view of \eqref{optimalstrategies'''} and \eqref{P,Q,u}, the optimal strategies are given by 
\begin{equation}\label{optimal_u}
\pi^*_t =\frac{\alpha-r}{\gamma\sigma^{2}}-\frac{u_{y}(t,\mathcal Y_t)}{u(t,\mathcal Y_t)}\quad \text{and}\quad \tilde{c}^{*}_t = \delta^{\psi}u(t,\mathcal Y_t)^{-\frac{\gamma\psi}{\theta}},
\end{equation}
where $u$ is the unique classical solution in Proposition~\ref{Prop: Parabolic}.  

\begin{remark}
	In the standard case with no default, since the bankruptcy level is absent, the first line of \eqref{PDE_Ex1} holds for all $y\in\R$ and the second line becomes the terminal condition $u(T,y)=1$. This simplified \eqref{PDE_Ex1}, via the ansatz $u(t,y) = p(t)$ for all $(t,y)\in[0,T]\times \R$, reduces to the ODE   
	\begin{equation}\label{ODE_Ex1}
		p'(t) + \frac{1-\gamma}{\gamma}\left(r + \frac{(\alpha-r)^{2}}{2\gamma\sigma^{2}}\right)p(t) -\frac{\delta\theta}{\gamma}p(t)+\theta\frac{\delta^{\psi}}{\gamma\psi}\left(p(t)^{1-\frac{\gamma\psi}{\theta}}\right)=0,\quad p(T)=1.
	\end{equation}
	This is a Bernoulli equation with the closed-form solution
	\begin{equation*}
		p(t) = \left(e^{-C^{*}(T-t)} - \frac{\delta^{\psi}}{C^{*}}\left(e^{-C^{*}(T-t)}-1\right)\right)^{\frac{\theta}{\gamma\psi}},\quad t\in[0,T],
	\end{equation*}
	with $C^{*} := \delta\psi + (1-\psi)\big(r + \frac{(\alpha-r)^{2}}{2\gamma\sigma^{2}}\big)$. Plugging $u(t,y)=p(t)$ into \eqref{optimal_u} then recovers the optimal consumption and investment policies in \cite[Theorem 4]{Schroder96}, i.e. 
	\begin{equation}\label{optimal_no default}
		\pi^{*}_t = \frac{\alpha-r}{\gamma\sigma^{2}}\quad \text{and}\quad \tilde{c}^*_t= C^{*}\left(\left(\frac{C^{*}}{\delta^{\psi}}-1\right)e^{-C^{*}(T-t)}+1\right)^{-1}.
			\end{equation}
\end{remark}

In the sequel, we fix $T=5\text{ (years)},\, L= 90,\, \kappa=0,\, r=\delta=0.03,\, \alpha = 0.12$, and  $\sigma =0.3$, the same parameters used in \cite{Kraft06}. With $\alpha>r$, the firm value $S$ is expected to outpace $S^0$ in the long term. Such upward potential is taken away, however, once $S$ falls to the level $L=90$ and the forced bankruptcy takes place. In other words, compared with the standard no-default case, the expected growth of $S$ is cut down due to the imposed bankruptcy level, making investment in $S$ less valuable. The agent is likely to save more (relative to consumption and investment), relying more on the riskfree asset $S^0$ to compensate for the loss of return in the risky one $S$. 

\begin{figure}[h!]
	\centering
	\includegraphics[width=\linewidth]{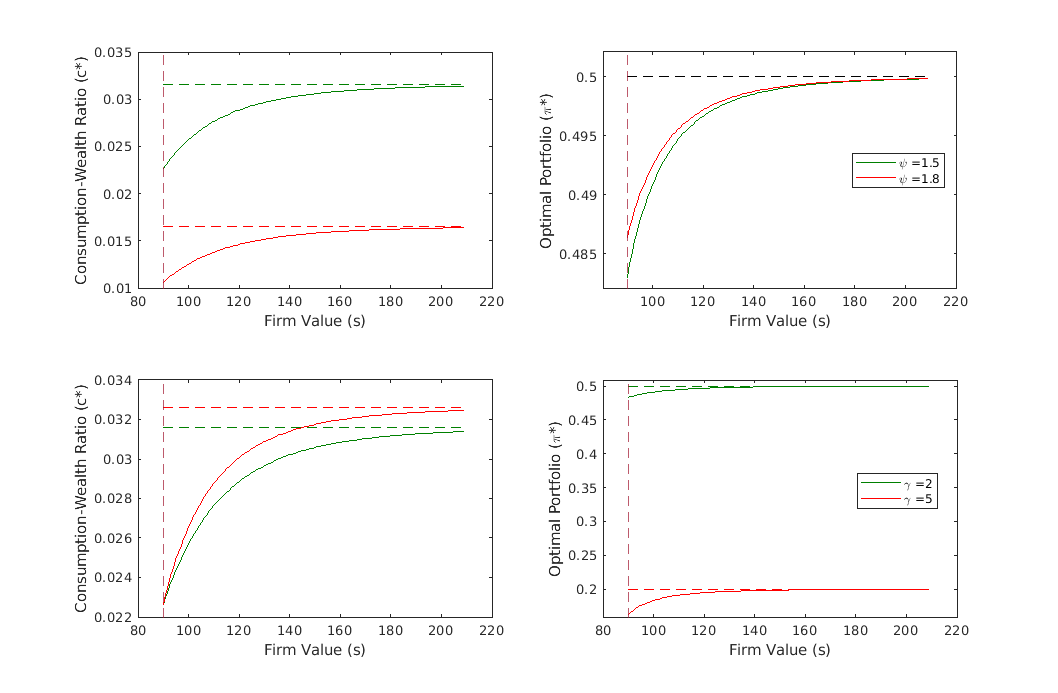}
	\caption{\small Optimal consumption and investment ratios at $t=0$ versus firm value. The upper panel fixes $\gamma=2$ and changes $\psi$; the lower one fixes $\psi=1.5$ and changes $\gamma$. The solid curves are computed via \eqref{optimal_u}; the dotted lines are benchmark levels in the no-default case, computed via \eqref{optimal_no default}. The upper-right plot has only one dotted line because the no-default level $\frac{\alpha-r}{\gamma\sigma^2}$ is independent of $\psi$.}
	\label{fig:cons.invest.default}
\end{figure}

Indeed, as shown in Figure \ref{fig:cons.invest.default}, for varying values of $\gamma$ and $\psi$, the optimal consumption and investment ratios (the solid curves) are both below the classical levels in the no-default case (the dotted lines), suggesting a larger proportion of wealth saved. Notably, the closer the firm value to the bankruptcy level, the more significant the departure from the no-default case.\footnote{This is in line with the finding in \cite{Kraft06} under time-separable utilities: for $\gamma>1$, the optimal investment strategy deviates below Merton's ratio $\frac{\alpha-r}{\gamma\sigma^{2}}$, more significantly as the firm value is closer to the bankruptcy level.} Intuitively, if the initial firm value $S_0$ is far above $L=90$, forced bankruptcy is unlikely and thus has little impact on the expected growth of $S$. The optimal strategies are then similar to those in the no-default case.\footnote{In fact, the optimal strategies converge to those in the no-default case as $S_0\to\infty$.} When $S_0$ is closer to $L=90$, forced bankruptcy is more probable,  distorting the optimal strategies more severely.

Intriguingly, the optimal investment ratio $\pi^*$ is sensitive to EIS $\psi$---in contrast to most findings under Epstein-Zin preferences. In the no-default (and thus fixed-horizon) case, under a constant investment opportunity, $\pi^*$ remains the standard Merton ratio $\frac{\alpha-r}{\gamma \sigma^2}$, independent of $\psi$; see \cite[Theorem 4]{Schroder96}. Even when the investment opportunity is stochastic, \cite[Section 3]{Xing} shows numerically that $\psi$ has little impact on $\pi^*$; see the discussion below \cite[(3.2)]{Xing}. The introduction of the bankruptcy level alters this markedly: $\pi^*$ now increases with $\psi$. 

Note that the bankruptcy level poses additional timing uncertainty---{\it When will the investment horizon end?} When the firm value is close to the bankruptcy level, it is highly probably that the horizon will end in near term. There is yet {\it risk} that the firm value rises again before touching the bankruptcy level, enlarging the horizon indefinitely. With $\gamma\psi>1$ (due to our specification $\gamma,\psi>1$), the agent prefers early resolution of uncertainty (see \cite{KP78, Skiadas98}), therefore desires a {\it risk premium} to compensate for the timing uncertainty. As $\psi$ increases, the corresponding rise in $\pi^*$ can be viewed, in a sense, as the added {\it risk premium}, under a stronger preference for early resolution of uncertainty. Indeed, if the risk actually materializes, i.e. $S$ rises again without touching $L=90$, a larger risky position yields a richer compensation.   


\subsection{The Exit Problem: Unbounded Time Horizon}\label{subsec:Heston}
In this example, we let the correlation between $W$ and $\hat W$ be constant, i.e. $\rho(t,y)\equiv \rho\in [-1,1]$, and take \eqref{EQ8} and \eqref{B&S} to be 
\begin{equation}\label{B&S'}
\begin{split}
dY_{t} &= -\alpha(Y_{t}-m^{2})dt + k\sqrt{Y_{t}}dW_{t},\quad Y_{0}=y>0,\\
dS^0_t &= r S^0_t dt,\quad S^0_0 =1,\\
dS_{t} &= S_{t}\left[\left(r+\lambda\cdot \big(Y_{t}+\eps\big)\right) dt + \sqrt{Y_{t}+\eps}dW_{t}^{\rho}\right],\quad S_{0}=s>0,
\end{split}
\end{equation}
where $\alpha$, $r$, $k$, $m$, $\lambda$ and $\eps$ are given nonnegative constants, with $2\alpha m^2>k^2$ satisfied such that $Y_t>0$ $\forall t\ge 0$ a.s. 
For $\eps=0$, this is the standard Heston model of stochastic volatility, which has been investigated in \cite{Kraft13, Kraft17, Xing}. For $\eps>0$, this forms an $\eps$-modification of the Heston model. 
It is used in \cite{Zawisza15} to obtain bounds of the market price of risk in the Scott and Stein-Stein models of stochastic volatility, as well as in \cite{Pham02}  to ensure uniform ellipticity in the Hull-White model. In the following, we will take $\rho=0$ (whence $\hat\rho=1$), so that $W^\rho = \hat W$. This is supported by empirical analysis in \cite{Silva03, Dragulescu02}, although there are other estimates of $\rho$ in the literature (such as \cite{Pan02, Chernov00}). 

As illuminated by the seminal work De Long et al.\ \cite{DeLong90}, a rational investor faces two distinct kinds of risk---the {\it fundamental risk} and the {\it noise trader risk}. The former refers to the inherent uncertainty in an asset's expected return (due to technological advances, macroeconomic conditions, etc.), while the latter is the risk that noise traders may {sustainably} push the asset price away from its fundamental value. This perspective differs from the classical arbitrage argument: Friedman \cite{Friedman53} and Fama \cite{Fama65}, among many others, assert that noise traders are met in the market by rational arbitrageurs who actively trade against them, which effectively drives the asset price back to normal. In reality, the opposite can happen: precisely in view of noise traders' unpredictable behavior (that may distort the asset price indefinitely), rational investors are less willing to trade; see e.g., \cite{DeLong90, SV97, Barberis98}. In a sense, noise traders intimidate rational investors into retreat. 

Bearing this in mind, we interpret \eqref{B&S'} as follows. Let $Y$ stand for an underlying economic factor. The expected return of $S$ (i.e., the term $r+\lambda\cdot(Y_{t}+\eps)$) then exemplifies the {\it fundamental risk}: the excess return beyond the riskfree rate $r$ is uncertain, influenced by $Y$ that evolves stochastically. The {\it noise trader risk}, on the other hand, is reflected through $\sqrt{Y_{t}+\eps}d\hat W_{t}$, which unpredictably alters the actual return of $S$ away from its mean. Here, we implicitly assume that the noise trader risk is more severe when $Y_t$ is larger, or equivalently, when the expected return of $S$ is larger. That is, a better economic condition (i.e., a larger $Y_t$) improves the fundamentals of $S$, which also raises the interest of noise traders in trading this stock. To better understand the sustained impact of noise traders on $S$, we introduce a process $\cW$ defined by
\begin{equation}\label{mean return}
d\cW_{t} = 
\sqrt{Y_{t}+\eps}d\hat{W}_{t},\quad \cW_{0}=w\in \R,
\end{equation}
which is the so-called zero-mean return process of $S$ in \cite{Masoliver08, Masoliver09}. 

For any $w\in \R$ and $y>0$, suppose that a rational investor considers the random horizon
\begin{equation}\label{tau_w,y}
\tau:= \inf\left\{t\ge 0 : (\cW_{t},Y_{t})\not\in\left(-\frac{L}{2},\frac{L}{2}\right)\times(y_1,y_2)\right\},
\end{equation}
where $L>0$ and $0<y_1<y_2$ are chosen by the investor a priori. Specifically, this investor exits the market when the sustained impact of noise traders on $S$, either positive or negative, is too significant (as depicted in \cite{DeLong90}) or the underlying economic factor reaches extreme values. In view of \eqref{B&S'} and \eqref{tau_w,y}, it is straightforward to check that Assumption~\ref{A1} is satisfied.

Set $\cD:=\left(-\frac{L}{2},\frac{L}{2}\right)\times(y_1,y_2)$. The PDE characterization established in Appendix~\ref{sec:Markov} (particularly Theorem~\ref{PDEapproach} and \eqref{representation}) suggests that a solution $(D,Z,\hat Z)$ to \eqref{EQ11} can be constructed by
\begin{align}\label{representation'}
	D_{t} = u(\cW_{t}^{w},Y_{t}^{y}),\quad Z_{t} = k\sqrt{Y_t}\ u_y(\cW_{t}^{w},Y_{t}^{y}),\quad \hat{Z}_{t}=\sqrt{Y_t+\eps}\ u_w(\cW_{t}^{w},Y_{t}^{y}),
\end{align}
where $u\in C^{2}(\overline{\cD})$ is the unique solution to the elliptic boundary value problem
\begin{equation}\label{PDE_Ex2}
	\begin{split}
		-\alpha(y-m^{2})u_{y}+\frac{k^{2}y}{2}u_{yy} + \frac{y+\eps}{2}u_{ww} = -G\left(y,u,k\sqrt{y}u_{y},\sqrt{y+\eps}u_{w}\right),\quad &\forall (w,y)\in\cD\\
		u(w,y) = 0,\quad &\forall (w,y)\in\partial\cD,
	\end{split}
\end{equation}
with
\begin{align}\label{Hpde}
	G(y,d,z,\hat z):=\frac{z^{2}}{2}+\frac{{\hat z}^{2}}{2\gamma} +\frac{1-\gamma}{\gamma} \lambda \sqrt{y+\eps}\hat z +\frac{\delta^{\psi}\theta}{\psi}e^{-\frac{\psi}{\theta}d} + (1-\gamma)\bigg(r + \frac{\lambda^2}{2\gamma}(y+\eps)\bigg)-\delta\theta.
\end{align}
In view of \eqref{optimalstrategies}, $\rho=0$, and $\hat\rho=1$, the formulas of the optimal strategies now take the form
\begin{equation}\label{optimal in u}
	\pi^*_t= \frac{\lambda+u_{w}(\cW_t,Y_t)}{\gamma}\qquad\hbox{and}\qquad \tilde{c}^*_t = \delta^\psi e^{-\frac{\psi}{\theta} u(\cW_t,Y_t)}.
\end{equation}
For these strategies to be optimal, we still need to check Assumption~\ref{A4}; recall Lemma~\ref{C4}. To this end, we need the following two observations and one technical lemma: 
\begin{itemize}
\item [(i)] $\tZ^2_t = \hat Z^2_t$ (as $\rho =0$ and $\hat\rho=1$) is bounded on $[0,\tau]$: By \eqref{representation'}, $\hat Z^2_t= (Y_t+\eps) u^2_w(\cW_t,Y_t)$ for $0\le t\le \tau$. As $u\in C^2(\overline \cD)$ and the domain $\cD$ is bounded, $\{\hat Z_t\}_{0\le t\le \tau}$ is bounded, for which an upper bound can be found by numerically solving \eqref{PDE_Ex2}.\footnote{We solve \eqref{PDE_Ex2} via finite element methods, using a triangular mesh with maximal edge length taken to be 0.005. Also, a suitable mollification of $\cD = \left(-\frac{L}{2},\frac{L}{2}\right)\times(y_1,y_2)$ is used to ensure the boundary regularity in Assumption \ref{A2}. }
\item [(ii)] $\widetilde C\le 0$: By applying the maximum principle in \cite{Barles95} to \eqref{PDE_Ex2}, we have $u\le 0$. Recalling that $D_t = u(\cW_t,Y_t)$ from \eqref{representation'}, we conclude $\widetilde C = \esssup\big(\sup_{t\ge 0} D_t\big) \le0 $.  
\end{itemize}

\begin{lemma}\label{lem:MGF}
For any $L>0$ and $(w,y)\in\cD$, consider $\tau_w:= \inf\left\{t\ge 0 : \cW_{t}\not\in\left(-\frac{L}{2},\frac{L}{2}\right)\right\}\ge \tau$. Then, $\E\left[e^{c\tau_w}\right]<\infty$ for all $c\in[0,c^*)$, with 
\begin{equation}\label{c^*}
c^*:= \left(\frac{\alpha m}{k}\right)^2\left(\sqrt{1+\left(\frac{k\pi}{\alpha L}\right)^{2}} -1\right) + \frac{\pi^{2}\eps}{2 L^{2}}>0.
\end{equation}
\end{lemma}
The proof of Lemma~\ref{lem:MGF} is relegated to Appendix~\ref{subsec:proof of Lemma MGF}.


In the following, we fix $\gamma = 2$, $\psi = 1.5$, $\delta = 0.08$, 
$r=0.05$, $\alpha=5$, $k^{2}=0.25$, $m^{2}=0.0225$, and  $\lambda =0.47$, the same parameters used in \cite{Kraft17, Xing}. Also, we take $y_1=0.001$ and $y_2=1$, leaving $L>0$ the only free variable. By the observations (i) and (ii) above and Lemma~\ref{lem:MGF}, as long as $L>0$ is small enough (so that $c^*$ in \eqref{c^*} is large enough), $\tau$ satisfies the two exponential moment conditions in Assumption \ref{A4}. Under our choice of parameters, numerical computations show that, with $\eps=0$, $\tau$ satisfies Assumption \ref{A4} for all $0<L\le 0.02$.\footnote{For $L>0.02$, Assumption \ref{A4} may still be satisfied by $\tau$, because our computation involves several upper bounds that need not be sharp.} 
An $\eps$-modification enlarges the allowable range of $L$. For instance, with $\eps=0.05$, $\tau$ satisfies Assumption \ref{A4} for all $0<L\le 0.08$. We plot the optimal strategies \eqref{optimal in u} on the state space of $(Y_t,\cW_t)$ in Figure \ref{fig:zeroEpsilon} (for $\eps=0$ and $L=0.02$) and Figure \ref{fig:positiveEpsilon} (for $\eps=0.05$ and $L= 0.02$, $0.08$).  

\begin{figure}[!htbp]
	\centering
	\includegraphics[width=\linewidth]{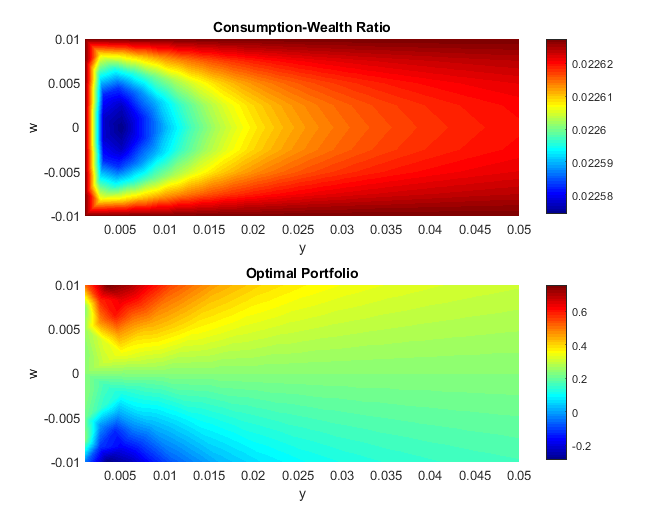}
	\setlength{\abovecaptionskip}{-25pt}
	\caption{\small Optimal strategies on the state space of $(Y_t,\cW_t)$ (for $\eps=0$ and $L=0.02$).}
	\label{fig:zeroEpsilon}
\end{figure}

\begin{figure}[!htbp]
	\centering
	\begin{subfigure}{.55\textwidth}
		\centering
		\includegraphics[width=\linewidth]{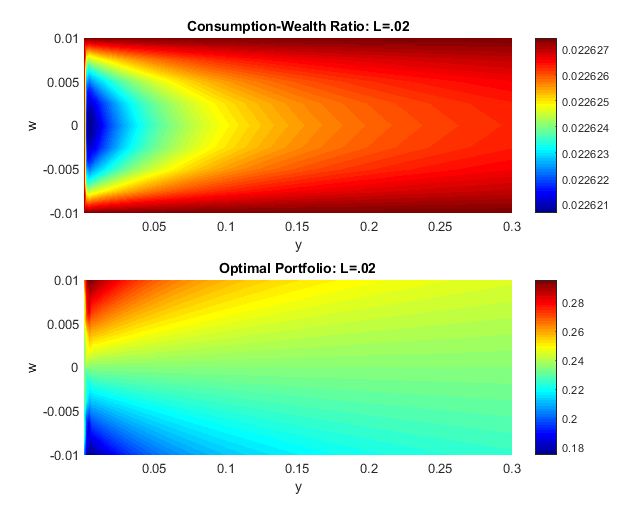}
		\caption{}
		\label{fig:optL02}
	\end{subfigure}%
	\begin{subfigure}{.55\textwidth}
		\centering
		\includegraphics[width=\linewidth]{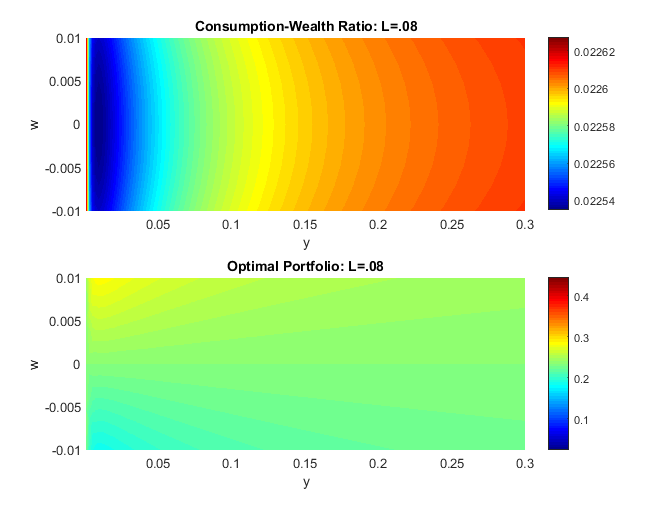}
		\caption{}
		\label{fig:optL08}
	\end{subfigure}
	\caption{\small Optimal strategies on the state space of $(Y_t,\cW_t)$ (for $\eps =0.05$ and $L=0.02$, $0.08$).}
	\label{fig:positiveEpsilon}
\end{figure}

The results differ largely from the classical ones on a fixed horizon. For a fixed horizon $T>0$, by \cite[Theorem 5.1]{Kraft17} and \cite[(2.14)]{Xing}, the optimal investment ratio under \eqref{B&S'} (with $\rho=0$ and $\hat\rho=1$) is
\[
\pi^*_{\text f} = \frac{\lambda_{t}}{\gamma\sigma_{t}^2} \equiv \frac{\lambda}{\gamma}.
\]
That is, it is optimal to keep a constant proportion of wealth in $S$, regardless of market conditions. Note that this can also be derived from \eqref{optimalstrategies}, as $\hat Z\equiv 0$ on a fixed horizon; see Remark~\ref{rem:hat Z=0} for details. By contrast, on the random horizon $\tau$, Figure \ref{fig:zeroEpsilon} shows that the optimal investment ratio $\pi^*$ changes continuously as the market evolves. To understand the various features of $\pi^*$ illustrated in Figure \ref{fig:zeroEpsilon}, we note from \eqref{B&S'} and \eqref{mean return} that given $0\le t<\tau$ with $\cW_t =w\in(-L/2,L/2)$, 
\begin{equation}\label{S expression}
S_{\tau} = S_t  \exp\big(A_{t,\tau}+(\cW_{\tau}-w)\big),\quad\hbox{with}\ A_{t,\tau} := {\int_t^{\tau}\big(r+(\lambda-1/2)(Y_s+\eps)\big) ds}
\end{equation}
By taking $0\le \eps<2/3$, we have $r+(\lambda-1/2)(Y_s+\eps)  = 0.05-0.03 (Y_s+\eps)\ge 0.02-0.03\eps>0$. Hence,  $A_{t,\tau}$ is strictly positive. 
\begin{itemize}
\item [1.] 
Suppose that $\cW_t=w\in(-L/2,L/2)$ is close to $L/2$.  
\begin{itemize}
\item [(a)] As it is likely that in near future the event $\{\cW_\tau = L/2\}$ will occur (so that $S_\tau = S_t e^{A_{t,\tau}+(L/2 -w)}>S_t$), holding more shares of $S$ at time $t$ can take advantage of this likely price increase from $S_t$ to $S_\tau$. Indeed, in Figure \ref{fig:zeroEpsilon}, as $w$ moves towards $L/2$ , $\pi^*$ is positive and increasing. 
\item [(b)] By \eqref{mean return}, a larger $Y_t=y$ implies a more volatile $\cW$, at least for a period right after time $t$. This enlarges the probability of $\{\cW_\tau = -L/2\}$, even though $\cW_t=w$ is currently near $L/2$. As $Y_t=y$ becomes smaller, $\cW$ is more stable and stays near $L/2$ for a longer period after $t$, which makes $\{\cW_\tau = L/2\}$ more probable. This strengthens the incentive to hold shares of $S$ at time $t$, as argued in (a). Indeed, in Figure \ref{fig:zeroEpsilon}, for a fixed $\cW_t=w$ closer to $L/2$, $\pi^*$ increases as $Y_t=y$ decreases. Intriguingly, $\pi^*$ attains its maximum near $y=0.005$ and then changes its behavior: as $y$ decreases further towards $y_1=0.001$, $\pi^*$ decreases. For $Y_t=y$ very close to $y_1$, it is likely that $Y$ reaches $y_1$ before $\cW$ reaches $L/2$. In such a case, one has $\cW_\tau<L/2$ and the argument in (a) no longer ensures $S_\tau>S_t$. The incentive to hold shares of $S$ is thus diminished, whence $\pi^*$ decreases. 
\end{itemize}
\item [2.] Suppose that $\cW_t=w\in(-L/2,L/2)$ is close to $-L/2$. 
\begin{itemize}
\item [(a)] As argued in 1(b), when $Y_t=y$ is large, $\cW$ is volatile. There is then a reasonable chance for $\{\cW_\tau = L/2\}$ to occur (so that $S_\tau = S_t e^{A_{t,\tau}+(L/2 -w)}>S_t$), even though $\cW_t=w$ is currently near $-L/2$. A moderate holding of $S$ therefore seems appropriate. As $Y_t=y$ decreases, $\cW$ becomes more stable and stays near $-L/2$ for a longer period after $t$, which makes $\{\cW_\tau = -L/2\}$ more probable. The incentive to hold shares of $S$ is thus diminished. Indeed, in Figure \ref{fig:zeroEpsilon}, for a fixed $\cW_t=w$ closer to $-L/2$, $\pi^*$ is positive but small when $Y_t=y$ is large, and decreases as $Y_t=y$ decreases. 
\item [(b)] When $Y_t=y$ is near 0.005, $\{\cW_\tau = -L/2\}$ becomes so probable and imminent that $S_\tau = S_t e^{A_{t,\tau}+(-L/2 -w)}<S_t$; more precisely, as $\tau$ is likely only slightly larger than $t$, $A_{t,\tau}>0$ is not positive enough to compensate for the negative term $-L/2 -w$. This encourages shorting the asset $S$, which is confirmed in Figure \ref{fig:zeroEpsilon}: $\pi^*$ attains its minimum, which is negative, when $Y_t=y$ is near 0.005. Intriguingly, as $Y_t=y$ further decreases towards $y_1=0.001$, $\pi^*$ increases. For $Y_t=y$ very close to $y_1$, it is likely that $Y$ reaches $y_1$ before $\cW$ reaches $-L/2$. In such a case, one has $\cW_\tau>-L/2$ and it is no longer clear if $S_\tau<S_t$ holds. The incentive to short $S$ is then diminished, whence $\pi^*$ increases. 
\end{itemize}
\end{itemize}

By \cite[Theorem 5.1]{Kraft17} and \cite[(2.14)]{Xing}, the optimal consumption ratio on a fixed horizon $T>0$ is given by $\tilde{c}_{\text f}^*(t,Y_t) = \delta^\psi e^{-\frac{\psi}{\theta} v(t,Y_t)}$, where $v$ is the solution to a Cauchy problem. Clearly, $\tilde{c}^*_{\text f}$ differs from $\tilde{c}^*$ in \eqref{optimal in u}, as $u$ and $v$ are solutions to different differential equations. 

\begin{remark}
It is of interest to derive optimal strategies under time-separable utilities on the same random horizon $\tau$ in \eqref{tau_w,y} and compare them with our findings in this subsection. This turns out to be challenging: classical studies under time-separable utilities don't readily accommodate $\tau$. They require a random horizon to either be bounded by a fixed $T>0$ (e.g., \cite{Karatzas00, Blanchet08, Monique15}), depend solely on tradable assets (e.g., \cite{Karatzas00}), or depend on nontradable assets in a stylized, ad-hoc manner (e.g., \cite{Blanchet08}). All these are violated by $\tau$ in \eqref{tau_w,y}. Also note that our numerical calculations hinge on the existence of a solution $(D,Z,\hat Z)$ to BSDE \eqref{EQ11}, established under $\gamma, \psi>1$ (Proposition~\ref{P5}). As our framework doesn't cover time-separable utilities (which require $\gamma=1/\psi$), it is unclear whether and how modifying our calculations can produce results in the time-separable case. 
\end{remark}


\appendix
\section{Proofs for Section \ref{EZ}}\label{sec:proofs for EZ} 
Let us first present a useful estimation and a fundamental result. 
Recall the spaces of processes introduced above Proposition~\ref{P1}. For any $T>0$, define the spaces $\mathcal{S}_{q}([0,T])$, $\mathcal{M}_{q}([0,T])$, and $\mathcal{B}_{q}([0,T])$ similarly, with $\tau$ replaced by $T$. 

For any $(Y_{t},Z_{t})\in\mathcal{B}_{2}([0,T])$ and $q>0$, by Burkh\"{o}lder-Davis-Gundy's inequality, there exists $K>0$, independent of $Y$, $Z$, and $T$, such that
\begin{align}
q\cdot \mathbb{E}\bigg[\sup_{t\in [0,T]}\int_{0}^{t} Y_{s} Z_{s}dB_{s}\bigg] &\le q K \mathbb{E}\bigg[(\int_{0}^{T}|Y_{s}|^{2}\|Z_{s}\|^{2}ds)^{1/2}\bigg] \notag\\
&\le \mathbb{E}\bigg[\Big(\sup_{0\le t\le T}|Y_{t}|^{2}\Big)^{1/2}\Big(q^2 K^2\int_{0}^{T}\|Z_{s}\|^{2}ds\Big)^{1/2}\bigg]\notag\\
&\le \frac{1}{2}\mathbb{E}\bigg[\sup_{t\in[0, T]}|Y_{t}|^{2}\bigg]+\frac{q^2K^2}{2}\E\bigg[\int_{0}^{T}\|Z_{s}\|^{2}ds\bigg]<\infty, \label{BDG}
\end{align}
where the third inequality follows from $ab\le \frac{a^2}{2}+\frac{b^2}{2}$ for all $a,b\in \R$, and the finiteness is due to $(Y_{t},Z_{t})\in\mathcal{B}_{2}([0,T])$.

\begin{lemma}\label{L2}
	Fix $T<\infty$. Given $(Y_{t},Z_{t})\in\mathcal{B}_{2}([0,T])$, the continuous local martingale $\int_{0}^{t} Y_{s}Z_{s}dB_{s}$, $t\in[0,T]$, is a uniformly integrable martingale.
\end{lemma}

\begin{proof}
	It suffices to show that $\mathbb{E}\big[\sup_{t\in [0,T]}\int_{0}^{t} Y_{s}Z_{s}dB_{s}\big]<\infty$, which is true by \eqref{BDG}. 
\end{proof}


\subsection{Proof of Proposition \ref{P1}}\label{subsec:proof of Proposition P1}
Motivated by \cite{Briand00} and \cite{Pardoux99}, we will construct a sequence of solutions that is Cauchy in $\mathcal{B}_{2}$, and show that its limit solves \eqref{EQ5}. In the rest of the proof, we set
\begin{equation}\label{xi}
\xi := e^{-\delta\theta\tau}c_{\tau}^{1-\gamma}\quad \hbox{and}\quad p := 1-{1}/{\theta}>1\ \hbox{(by $\gamma,\psi>1$)}.
\end{equation}

{\bf Step 1: Construct a sequence of solutions $(Y^n,Z^n)_{n\in \N}$ in $\B_2$.} 
For each $n\in\N$, we aim to construct a solution $(Y^n_{t},Z^n_{t})_{t\ge 0}\in \B_2$ to the BSDE
\begin{equation}\label{BSDE t>0}
Y_{t}^{n} = \xi + \int_{t\wedge\tau}^{\tau}\ind_{[0,n]}(s)F(s,c_{s},Y^n_s)ds - \int_{t\wedge \tau}^{\tau}Z_{s}^{n}dB_s,\quad t\ge 0.
\end{equation}
Note that \eqref{BSDE t>0} is a random-horizon BSDE, and a solution $(Y^n_{t},Z^n_{t})_{t\ge 0}$ will be constructed by finite-horizon results in \cite{Xing} and a proper extension to the random horizon. Specifically, for the fixed time horizon $[0,n]$, thanks to the construction in \cite[Proposition 2.2]{Xing}, there exists a unique solution $(Y_{t},Z_{t})_{t\in[0,n]}\in\B_2([0,n])$ to the fixed-horizon BSDE
\begin{equation}\label{BSDE [0,n]}
Y_{t} = \E[\xi\mid\F_n] + \int_{t}^{n}\ind_{[0,\tau]}(s)F(s,c_{s},Y_s)ds - \int_{t}^{n}Z_{s}dB_s,\quad t\in[0,n].
\end{equation}
Specifically, $Y$ is continuous, $0<Y_t \le \E[\xi\mid \F_t]$ a.s. for all $t\in [0,n]$ (hence, $Y$ is of class $D$).\footnote{In general, the solution derived in \cite[Proposition 2.2]{Xing} need not lie in $\B_2([0,n])$. 
	This is because \cite{Xing} assumes only integrability on the terminal condition $\xi$, instead of square-integrability. In our case, as $c\in\cC$, 
	$
	\E\left[\E[\xi|\mathcal{F}_{n}]^2\right]\le \E[\xi^{2}]= \E[e^{-2\delta\theta\tau}c_{\tau}^{2(1-\gamma)}]<\infty.
	$
	With $\E[\xi|\mathcal{F}_{n}]$ being square integrable, the construction in \cite[Proposition 2.2]{Xing} then yields a solution in $\B_2([0,n])$. Specifically, we obtain from Step 1 of its proof (see \cite[Appendix A]{Xing}) the desired solution, with no need of Step 2 therein.} 
On the other hand, thanks to $c\in\cC$, \eqref{C} implies $\E[\xi^{2}]= \E[e^{-2\delta\theta\tau}c_{\tau}^{2(1-\gamma)}]<\infty$, i.e. $\xi$ is square integrable. Thus, by the martingale representation theorem, there exists $\eta\in\M_2$ such that 
\begin{equation}\label{eta}
\E[\xi\mid\F_t] = \xi - \int_t^\tau \eta_s dB_s\quad \hbox{and}\quad \eta_t = 0\ \hbox{for $t> \tau$}.  
\end{equation}
Now, define $(Y^n_{t},Z^n_{t})_{t\ge 0}$ as follows: $(Y^n_{t},Z^n_{t}) := (Y_{t},Z_{t})$ for $0\le t\le n$, and $Y_{t}^{n} := \mathbb{E}[\xi|\mathcal{F}_{t}]$ and $Z_{t}^{n} := \eta_t$ for all $t>n$. By \eqref{BSDE [0,n]} and \eqref{eta}, it can be checked directly that $(Y^n_{t},Z^n_{t})_{t\ge 0}\in \B_2$ is a solution to \eqref{BSDE t>0}; a similar construction can be found in \cite[Theorem 4.1]{Pardoux99}. 

{\bf Step 2: Show that the sequence $(Y^n,Z^n)_{n\in\N}$ is Cauchy in $\B_2$.} For any $m,n\in\N$ with $m>n$, consider $\Delta Y_{t} := Y_{t}^{m}-Y_{t}^{n}$, $\Delta Z_{t} := Z_{t}^{m}-Z_{t}^{n}$, and $\Delta F(t,c_t,Y_t)  := F(t,c_t,Y^{m}_t)-F(t,c_t,Y^{n}_t)$. We intend to show that $\|(\Delta Y_{t},\Delta Z_{t})\|_{\B_{2}}\to 0$ as $m,n\to\infty$. 

For $0\le t\le n$, observe from \eqref{BSDE t>0} that
\begin{align}
\Delta Y_{t} &=  \int_{t\wedge\tau}^{n\wedge\tau}\Delta F(s,c_{s},Y_{s})ds - \int^{\tau}_{t\wedge\tau}\Delta Z_{s}dB_{s} + \int_{n\wedge\tau}^{m\wedge\tau}F(s,c_s,Y_{s}^{m})ds \notag\\
&=\Delta Y_{n\wedge \tau} + \int_{t\wedge\tau}^{n\wedge\tau}\Delta F(s,c_{s},Y_{s})ds - \int_{t\wedge\tau}^{n\wedge\tau}\Delta Z_{s}dB_{s}.  \label{EQ7'}
\end{align}
Recall $p>1$ in \eqref{xi}. Applying It\^{o}'s formula to $|\Delta Y_{t}|^{2}$, with $\Delta Y_t$ as in \eqref{EQ7'}, yields
\begin{align}
&|\Delta Y_{t\wedge \tau}|^{2} +\int_{t\wedge\tau}^{n\wedge\tau}\|\Delta Z_{t}\|^{2}ds\notag\\
&\hspace{0.2in}  =|\Delta Y_{n\wedge \tau}|^{2} + 2\int_{t\wedge\tau}^{n\wedge\tau}\Delta Y_{s} \Delta F(s,c_{s},Y_{s}) ds  -2\int_{t\wedge\tau}^{n\wedge\tau} \Delta Y_{s} \Delta Z_{s} dB_{s} \notag\\
&\hspace{0.2in}= |\Delta Y_{n\wedge\tau}|^{2} + \int_{t\wedge\tau}^{n\wedge\tau}(2\delta\theta\Delta Y_{s} e^{-\delta s}c_{s}^{1-{1}/{\psi}}((Y_{s}^{m})^{p}-(Y_{s}^{n})^{p}))ds - 2\int_{t\wedge\tau}^{n\wedge\tau} \Delta Y_{s} \Delta Z_{s}dB_{s}.\label{EQ6'}
\end{align}
Since $\Delta Y_s = Y^m_s-Y^n_s$, the sign of $\Delta Y_s$ must be the same as that of $(Y_{s}^{m})^{p}-(Y_{s}^{n})^{p}$. This, together with $\theta<0$, gives $2\delta\theta\Delta Y_{s} e^{-\delta s}c_{s}^{1-{1}/{\psi}}((Y_{s}^{m})^{p}-(Y_{s}^{n})^{p})\le 0$. We then conclude from \eqref{EQ6'} that
\begin{equation}\label{esti 1}
|\Delta Y_{t\wedge\tau}|^{2} + \int_{t\wedge\tau}^{n\wedge\tau}\|\Delta Z_{t}\|^{2}ds   \le |\Delta Y_{n\wedge\tau}|^{2}  - 2\int_{t\wedge\tau}^{n\wedge\tau}  \Delta Y_{s} \Delta Z_{s}dB_{s}.
\end{equation}
This, together with Lemma~\ref{L2}, gives
\[
\E\bigg[\int_{0}^{n\wedge\tau}\|\Delta Z_{t}\|^{2}ds\bigg] \le \E[|\Delta Y_{n\wedge\tau}|^{2}]- 2\E\bigg[\int_{0}^{n\wedge\tau}  \Delta Y_{s} \Delta Z_{s}dB_{s}\bigg] = \E[|\Delta Y_{n\wedge\tau}|^{2}].
\]
Moreover, by using \eqref{esti 1} and \eqref{BDG}, with $q=2$, 
\begin{equation*}
\E\bigg[\sup_{0\le t\le n}|\Delta Y_{t\wedge\tau}|^{2}\bigg] \le \E\left[|\Delta Y_{n\wedge\tau}|^{2}\right] + \frac{1}{2}\mathbb{E}\bigg[\sup_{0\le t\le n}|\Delta Y_{t\wedge\tau}|^{2}\bigg] +  2 K^2\E\left[\int_{0}^{n\wedge\tau}\|\Delta Z_{s}\|^2 ds\right],
\end{equation*}
for some $K>0$, independent of $m$ and $n$. By the previous two inequalities, there exists $K_1>0$, independent of $m$ and $n$, such that 
\begin{equation}\label{Cauchy 1}
\mathbb{E}\bigg[\sup_{0\le t\le n}|\Delta Y_{t\wedge\tau}|^{2}+\int_{0}^{n\wedge\tau}\|\Delta Z_{t}\|^{2}ds\bigg] \le  K_1 \mathbb{E}[|\Delta Y_{n\wedge\tau}|^{2}].
\end{equation}

Next, for $n<t\le m$, observe from \eqref{BSDE t>0} that
\begin{equation*}
\Delta Y_{t} =  \int_{t\wedge\tau}^{m\wedge\tau}F(s,c_{s},Y_{s}^{m})ds - \int_{t\wedge\tau}^{\tau}\Delta Z_{s}dB_{s} = \int_{t\wedge\tau}^{m\wedge\tau}F(s,c_{s},Y_{s}^{m})ds - \int_{t\wedge\tau}^{m\wedge \tau}\Delta Z_{s}dB_{s}, 
\end{equation*}
where the second equality follows from $\int_{m\wedge\tau}^{\tau}\Delta Z_{s}dB_{s}=0$, as $Z^m_s = Z^n_s=\eta_s$ for all $s> m\wedge \tau$ by the construction in Step 1. 
Applying It\^{o}'s formula to $|\Delta Y_{t}|^{2}$, with $\Delta Y_t$ as above, gives
\begin{align}
|\Delta Y_{t\wedge\tau}|^{2} &+ \int_{t\wedge\tau}^{m\wedge\tau}\|\Delta Z_{s}\|^{2}ds = 2\int_{t\wedge\tau}^{m\wedge\tau} \Delta Y_{s} F(s,c_{s},Y_{s}^{m}) ds - 2\int_{t\wedge\tau}^{m\wedge\tau}\Delta Y_{s} \Delta Z_{s}dB_{s} \notag\\
&\le 2\delta|\theta| \int_{t\wedge\tau}^{m\wedge\tau}e^{-\delta s} c_s^{1-1/\psi} (\E[\xi\mid\F_s])^{p+1} ds - 2\int_{t\wedge\tau}^{m\wedge\tau}\Delta Y_{s} \Delta Z_{s}dB_{s}, 
\label{lala}
\end{align}
where the inequality follows from
\begin{align*}
\Delta Y_{t} F(s,c_{s},Y_{s}^{m}) &= \delta\theta e^{-\delta s} c_s^{1-1/\psi} (Y_s^m)^p (Y_s^m-Y_s^n)\\
&\le -\delta\theta e^{-\delta s} c_s^{1-1/\psi} (Y_s^m)^p Y_s^n \le -\delta\theta e^{-\delta s} c_s^{1-1/\psi} (\E[\xi\mid\F_s])^{p+1},
\end{align*}
thanks to $0\le Y^m_s, Y^n_s\le \E[\xi\mid\F_s]$ and $\theta<0$. Observe that 
\begin{align}
&\mathbb{E}\bigg[\int_{n\wedge\tau}^{m\wedge\tau} e^{-\delta s}c_{s}^{1-\frac{1}{\psi}}(\mathbb{E}[\xi|\mathcal{F}_{s}])^{p+1}ds\bigg]\notag\\
&\hspace{0.2in}\le \E\left[\left(\int_{n\wedge\tau}^{m\wedge\tau}\mathbb{E}[\xi|\mathcal{F}_{s}]^{2(p+1)}ds\right)^{\frac{1}{2}}\left(\int_{n\wedge\tau}^{m\wedge\tau} e^{-2\delta s}c_{s}^{2(1-\frac{1}{\psi})}ds\right)^{\frac{1}{2}}\right]\notag\\
&\hspace{0.2in}\le \E\left[\int_{n\wedge\tau}^{m\wedge\tau}\mathbb{E}[\xi^{p+1}|\mathcal{F}_{s}]^{2}ds\right]^{\frac{1}{2}}\E\left[\int_{n\wedge\tau}^{m\wedge\tau} e^{-2\delta s}c_{s}^{2(1-\frac{1}{\psi})}ds\right]^{\frac{1}{2}}, \label{abu}
\end{align}
where the first inequality follows from applying the Cauchy-Schwarz inequality to the integral inside the expectation, and the second inequality follows from applying the Cauchy-Schwarz inequality to the expectation and then using Jensen's inequality. By \eqref{xi} and the fact that $c\in\cC$,
\begin{equation}\label{L^2}
\E[\xi^{2(p+1)}]= \E\left[e^{-2(p+1)\delta\theta\tau}c_{\tau}^{2(p+1)(1-\gamma)}\right]<\infty.
\end{equation}
Thus, the martingale representation theorem entails $\mathbb{E}[\xi^{p+1}|\mathcal{F}_{s}] = \E[\xi^{p+1}] + \int_0^{t\wedge\tau} \nu_s dB_s$ for some adapted process $\nu$. This allows the use of \cite[Lemma 4.1]{Darling97}, which yields the finiteness of 
$C_{1} := \E\big[\int_{0}^{\tau}e^{2\delta s}\E\left[\xi^{p+1}|\F_{s}\right]^{2}ds\big]^{{1}/{2}}$. 
We then obtain from \eqref{abu} that 
\begin{equation}\label{lalahoho}
\mathbb{E}\bigg[\int_{n\wedge\tau}^{m\wedge\tau} e^{-\delta s}c_{s}^{1-\frac{1}{\psi}}(\mathbb{E}[\xi|\mathcal{F}_{s}])^{p+1}ds\bigg]\le C_{1}\E\left[\int_{n\wedge\tau}^{m\wedge\tau} e^{-2\delta s}c_{s}^{2(1-\frac{1}{\psi})}ds\right]^{\frac{1}{2}}.
\end{equation}
In view of \eqref{lala} and Lemma \ref{L2}, this directly implies
\begin{align*}
\E\left[\int_{n\wedge \tau}^{m\wedge\tau}\|\Delta Z_{s}\|^{2}ds\right] &\le 2\delta|\theta|C_{1}\E\left[\int_{n\wedge\tau}^{m\wedge\tau} e^{-2\delta s}c_{s}^{2(1-\frac{1}{\psi})}ds\right]^{\frac{1}{2}}.
\end{align*}
Moreover, by using \eqref{lala} and \eqref{BDG}, with $q=2$, 
\begin{align*}
\E\bigg[\sup_{n\le t\le m}|\Delta Y_{t\wedge\tau}|^{2}\bigg] &\le 2\delta|\theta| \E\left[ \int_{n\wedge\tau}^{m\wedge\tau}e^{-\delta s} c_s^{1-1/\psi} (\E[\xi\mid\F_s])^{p+1} ds\right]\\
&\hspace{0.2in} + \frac{1}{2}\mathbb{E}\bigg[\sup_{n\le t\le m}|\Delta Y_{t\wedge\tau}|^{2}\bigg] +  2 K^2\E\left[\int_{n\wedge\tau}^{m\wedge\tau}\|\Delta Z_{s}\|^2 ds\right],
\end{align*}
for some $K>0$, independent of $m$ and $n$. By combining the previous two inequalities and using \eqref{lalahoho}, there exists $K_2>0$, independent of $m$ and $n$, such that 
\begin{align}\label{Cauchy 2}
&\E\left[\sup_{n\le t\le m}|\Delta Y_{t\wedge\tau}|^{2}+\int_{n\wedge\tau}^{m\wedge\tau}\|\Delta Z_{t}\|^{2}ds\right]\le K_2 \delta|\theta| \E\left[\int_{n\wedge\tau}^{m\wedge\tau} e^{-2\delta s}c_{s}^{2(1-\frac{1}{\psi})}ds\right]^{\frac{1}{2}}.
\end{align}	
By \eqref{Cauchy 1}, \eqref{Cauchy 2}, and recalling that $Y^m_s=Y^n_s$ and $Z^m_s=Z^n_s$ for all $s> m\wedge \tau$, we have 
\begin{align}
\E&\left[\sup_{ t\ge 0}|\Delta Y_{t\wedge\tau}|^{2}+\int_0^{\tau}\|\Delta Z_{t}\|^{2}ds\right] 
\le K_1\E[|\Delta Y_{n\wedge\tau}|^{2}]+ K_2\delta|\theta|\E\left[\int_{n\wedge\tau}^{m\wedge\tau} e^{-2\delta s}c_{s}^{2(1-\frac{1}{\psi})}ds\right]^{\frac{1}{2}}.\label{Cauchy 3}
\end{align}
We know from Step 1 that $0\le Y^m_{n\wedge\tau}, Y^n_{n\wedge\tau}\le \E[\xi\mid \F_{n\wedge\tau}]$, which imply
\begin{equation}\label{Delta Y}
|\Delta Y_{n\wedge\tau}|^{2} = |Y^m_{n\wedge\tau}-Y^n_{n\wedge\tau}|^2 \le (Y^m_{n\wedge\tau})^2 + (Y^n_{n\wedge\tau})^2 \le 2\E[\xi\mid\F_{n\wedge\tau}]^2 \le 2\E[\xi^2\mid\F_{n\wedge\tau}].
\end{equation}
Since \eqref{L^2} implies that $\{\E[\xi^2\mid\F_k]:k\ge 0\}$ is uniformly integrable, we deduce from \eqref{Delta Y} that
\[
\lim_{m,n\to\infty}\E[|\Delta Y_{n\wedge\tau}|^{2} ] = \lim_{m,n\to\infty} \E\left[|Y^m_{n\wedge\tau}-Y^n_{n\wedge\tau}|^2\right]= \E\left[\lim_{m,n\to\infty}|Y^m_{n\wedge\tau}-Y^n_{n\wedge\tau}|^2\right]=\E\left[|\xi-\xi|^2\right] =0.
\]
Finally, thanks to \eqref{C}, the second term in \eqref{Cauchy 3} vanishes as $m,n\to\infty$. Therefore, we conclude from \eqref{Cauchy 3} that $\|(\Delta Y,\Delta Z)\|_{\B_{2}}\to 0$ as $m,n\to\infty$, i.e. $\{(Y^n,Z^n)\}_{n\in\N}$ is Cauchy in $\B_2$. Since $\mathcal{B}_{2}$ is complete\footnote{This is a standard result that can be found in  \cite{Briand00} and \cite{Pardoux99}. It is used explicitly in  \cite{Briand00} and \cite{Pardoux99}, as well as implicitly in the proofs of \cite{Darling97} and \cite{DuffieLions}.}, $\lim\limits_{n\rightarrow\infty}  (Y^{n},Z^{n}) = (Y,Z) \in \mathcal{B}_{2}$ exists.	

{\bf Step 3: The limit $(Y,Z)\in\mathcal B_2$ solves \eqref{EQ5}.} 
For any $n\in\N$, since $(Y^{n}_t,Z^{n}_t)_{t\ge 0}\in\B_2$ solves \eqref{BSDE t>0},
\begin{align}
Y_{t}^{n} &= \xi+ \int_{t\wedge\tau}^{\tau}\mathbbm{1}_{[0,n]}(s)F(s,c_s,Y_{s}^{n})ds - \int_{t\wedge\tau}^{\tau}Z_{s}^{n}dB_s,  \quad t\ge 0. \label{Soln Pre-limit}
\end{align}
We intend to prove that each term in \eqref{Soln Pre-limit} converges to a corresponding term in \eqref{EQ5} $\forall t\ge 0$ $\P$-a.s., as $n\to\infty$. This then implies that $(Y,Z)$ satisfies \eqref{EQ5} $\forall t\ge 0$ $\P$-a.s., as desired. 

First, $Y^n\to Y$ in $\cS_2$ already implies $Y^n_t\to Y_t\ \forall t\ge 0$ $\P$-a.s. To show that $\int_{t\wedge\tau}^{\tau} \mathbbm{1}_{[0,n]}(s) F(s,c_{s},Y^n_{s}) ds\to \int_{t\wedge\tau}^{\tau}F(s,c_{s},Y_{s})ds$ $\forall t\ge 0$ $\P$-a.s. (possibly up to a subsequence), it suffices to prove that
\begin{align*}
\E&\left[\sup_{0\le t<\infty} \int_{t\wedge \tau}^{\tau}|\mathbbm{1}_{[0,n]}(s)F(s,c_{s},Y^n_{s})-F(s,c_{s},Y_{s})|ds\right] \\
&= \E\left[\int_{0}^{\tau}|\mathbbm{1}_{[0,n]}(s)F(s,c_{s},Y^n_{s})-F(s,c_{s},Y_{s})|ds\right]\to 0\quad \hbox{as $n\to\infty$},
\end{align*}
which is equivalent to $\mathbbm{1}_{[0,n]}(\cdot)F(\cdot,c_{\cdot},Y^n_{\cdot})\to F(\cdot,c_{\cdot},Y_{\cdot})$ in $L^1(\mu)$, for the measure $\mu:=\mathbbm{1}_{\{0\le t \le\tau\}} dt\times d\mathbb{P}$. Since $Y^{n}_t\to Y_t$ $\forall t\ge 0$ $\P$-a.s., the continuity of $F$ implies $\mathbbm{1}_{[0,n]}(t)F(t,c_t,Y^{n}_t)\rightarrow F(t,c_t,Y_t)$ $\forall t\ge 0$ $\mathbb{P}$-a.s. Also, in view of $0\le Y_{s}^{n}, Y_s\le \E\left[\xi|\F_{s}\right]$ for all $s\ge 0$ and $n\in\N$, 
\[
|F(s,c_s,Y^n_s)| \le \delta |\theta| e^{-\delta s} c_s^{1-1/\psi} \E[\xi\mid\F_s]^p\quad \forall s\ge 0\ \hbox{and}\ n\in\N\cup\{0\}, 
\]
with $Y^0:= Y$. Hence, if we can show that $e^{-\delta \cdot} c_{\cdot}^{1-1/\psi} \E[\xi\mid\F_\cdot]^p$ is $\mu$-integrable, the dominated convergence theorem will give the desired convergence $\mathbbm{1}_{[0,n]}(\cdot)F(\cdot,c_{\cdot},Y_{\cdot})\to F(\cdot,c_{\cdot},Y_{\cdot}^{n})$ in $L^1(\mu)$. To this end, observe that
\begin{align}
\E\left[\int_{0}^{\tau} e^{-\delta s} c_s^{1-1/\psi} \E[\xi\mid\F_s]^p ds \right] &\le\E\left[\int_{0}^{\tau} e^{-2\delta s}c_{s}^{2(1-\frac{1}{\psi})}ds\right]^{\frac{1}{2}}\E\left[ \int_{0}^{\tau}  \E[\xi^p\mid\F_s]^{2} ds\right]^{\frac{1}{2}},\label{abudu}
\end{align}
where the first inequality follows from applying the Cauchy-Schwartz inequality twice and then the Jensen inequality (similarly to \eqref{abu}). Note that $\E\big[\int_{0}^{\tau} e^{-2\delta s}c_{s}^{2(1-{1}/{\psi})}ds\big]<\infty$, as $c\in\cC$; see \eqref{C}. By the arguments similar to \eqref{L^2} and the discussion below it, we get $\E\left[ \int_{0}^{\tau} \E[\xi^p\mid\F_s]^{2} ds\right]<\infty$. We then conclude from \eqref{abudu} the $\mu$-integrability of $e^{-\delta \cdot} c_{\cdot}^{1-1/\psi} \E[\xi\mid\F_\cdot]^p$, as desired.  

It remains to show that $\int_{t\wedge\tau}^{\tau} Z_{s}^{n}dB_{s}\to \int_{t\wedge\tau}^{\tau}Z_{s}dB_{s}$ $\forall t\ge 0$ $\P$-a.s. Let $\Delta Z_s^n := Z^n_s - Z_s$. Since $Z^n\to Z$ in $\M_2$, by the It\^{o} isometry,  
\begin{equation*}
\mathbb{E}\bigg[\left(\int_{t\wedge\tau}^{\tau}\Delta Z_{s}^{n}dB_{s}\right)^{2}\bigg]\le \E\left[\int_{0}^{\tau}\|\Delta Z_{s}^{n}\|^{2}ds\right] = \|\Delta Z^{n}\|^2_{\M_2} \rightarrow 0,\quad \hbox{for each $t\ge 0$}.
\end{equation*}
This implies $\int_{t\wedge\tau}^{\tau} Z_{s}^{n}dB_{s}\to \int_{t\wedge\tau}^{\tau}Z_{s}dB_{s}$ in probability, for each $t\ge 0$. Because every other term in \eqref{Soln Pre-limit} (either on the left or right hand side) converges $\forall t\ge0$ $\P$-a.s., $\int_{t\wedge\tau}^{\tau}Z_{s}^{n}dB_s$ must also converge $\forall t\ge 0$ $\P$-a.s. Hence, we have $\int_{t\wedge\tau}^{\tau} Z_{s}^{n}dB_{s}\to \int_{t\wedge\tau}^{\tau}Z_{s}dB_{s}$ $\forall t\ge 0$ $\P$-a.s.  


{\bf Step 4: $(Y,Z)$ is the unique solution to \eqref{EQ5} in $\B_2$.} 
Let $(Y',Z')$ be another solution to \eqref{EQ5} in $\B_2$ and set $(\Delta Y_{t}, \Delta Z_{t}) = (Y_t-Y'_t,Z_t-Z'_t)$ for all $t\ge 0$. By a calculation similar to \eqref{EQ6'} and following the arguments below it, we can achieve \eqref{Cauchy 1} in the same way. As $n\to\infty$ in \eqref{Cauchy 1}, by the monotone convergence theorem (applicable for the left-hand side), the dominated convergence theorem (applicable for the right-hand side, as $Y, Y'\in \mathcal S_2$), and $Y_\tau = Y'_\tau =\xi$, we get $\mathbb{E}\left[\sup_{0\le t<\infty}|\Delta Y_{t\wedge\tau}|^{2}+\int_{0}^{\tau}\|\Delta Z_{t}\|^{2}ds\right] \le 0$. That is, $(Y',Z')$ coincides with $(Y,Z)$ in $\B_2$.

\subsection{Proof of Theorem \ref{thm:EZ exists}}\label{subsec:proof of Theorem EZ exists}

By Proposition~\ref{P1}, a direct calculation shows that $(V^c, Z^c)$ is the unique solution to \eqref{EQ3} in $\B_2$, and $V^c$ is of class $D$. 
To show the last assertion that $V^c$ satisfies \eqref{EQ2} a.s., we first note that $t\mapsto V^c_{t}+ \int_{0}^{t\wedge\tau}f(c_{s},V^c_{s})ds$ is a martingale. Indeed, For any $0\le u\le t$,
\begin{align*}
\E_u&\left[V^c_{t}+ \int_{0}^{t\wedge\tau}f(c_{s},V^c_{s})ds\right] = \E_u\bigg[\frac{c^{1-\gamma}_\tau}{1-\gamma}+ \int_{0}^{\tau}f(c_{s},V^c_{s})ds - \int_{t\wedge\tau}^{\tau} Z^c_s dB_s\bigg]\\
&=  \int_{0}^{u\wedge\tau}f(c_{s},V^c_{s})ds + \E_u\bigg[\frac{c^{1-\gamma}_\tau}{1-\gamma}+ \int_{u\wedge\tau}^{\tau}f(c_{s},V^c_{s})ds - \int_{u\wedge\tau}^{\tau} Z^c_s dB_s + \int_{u\wedge\tau}^{t\wedge\tau} Z^c_s dB_s\bigg]\\
&= \int_{0}^{u\wedge\tau}f(c_{s},V^c_{s})ds + V^c_u,
\end{align*}	
where the last equality follows from $Z^c\in \M_2$. Fix $t\ge 0$. By the above martingale property,
\begin{align*}
V^c_{t} &= \mathbb{E}_{t}\bigg[V^c_{m} + \int_{t\wedge \tau}^{m\wedge\tau}f(c_{s},V^c_{s})ds\bigg],\quad \forall m\ge t.
\end{align*}
Thanks to the definition of $f$ in \eqref{f}, the above equality can be rewritten as
\begin{equation*}
V^c_{t}+\delta\theta\mathbb{E}_{t}\bigg[\int_{t\wedge \tau}^{m\wedge\tau}V^c_{s}ds\bigg]=\mathbb{E}_{t}[V^c_{m}] + \E_t\bigg[ \int_{t\wedge\tau}^{m\wedge\tau}\delta\frac{c_{s}^{1-{1}/{\psi}}}{1-\frac{1}{\psi}}\big((1-\gamma)V^c_{s}\big)^{1-\frac{1}{\theta}}ds\bigg].
\end{equation*}
As $m\to\infty$, similarly to \cite[(A.5)]{Xing}, since $V^c_t = e^{\delta\theta t} Y_t/(1-\gamma)\le 0$ for all $t\ge 0$, we may apply the monotone convergence theorem to the expectation on the left-hand side and the second expectation on the right-hand side; since $V^c$ is of class $D$, we may apply the dominated convergence theorem to the first expectation on the right-hand side. This ultimately yields
\begin{equation}\label{EQ2'}
V^c_{t}+\delta\theta\mathbb{E}_{t}\bigg[\int_{t\wedge \tau}^{\tau}V^c_{s}ds\bigg]=\mathbb{E}_{t}\bigg[V^c_{\tau} + \int_{t\wedge\tau}^{\tau}\delta\frac{c_{s}^{1-{1}/{\psi}}}{1-\frac{1}{\psi}}\big((1-\gamma)V^c_{s}\big)^{1-\frac{1}{\theta}}ds\bigg].
\end{equation}
To show that the conditional expectation $\mathbb{E}_{t}\left[\int_{t\wedge \tau}^{\tau}V^c_{s}ds\right]$ above is well-defined, note that $0\ge V^c_s = \frac{e^{\delta\theta s}}{1-\gamma} Y_s \ge \frac{1}{1-\gamma} \E[\xi\mid\F_s]$ for all $s\ge 0$, and thus
$
0\ge \E_t\left[\int_{t\wedge \tau}^{\tau}V^c_{s}ds\right] \ge \frac{1}{1-\gamma} \E_t\left[\int_{t\wedge \tau}^{\tau}\E[\xi\mid\F_s] ds\right]>-\infty,
$
where the finiteness in the last inequality follows from an argument similar to \eqref{L^2} and the discussion below it. Finally, observe that \eqref{EQ2'} readily gives \eqref{EQ2}. 


\section{Proofs for Section \ref{CI}}\label{sec:proofs for CI}

\subsection{Derivation of Proposition~\ref{P5}}\label{subsec:proof of Proposition P5}
As discussed above Proposition~\ref{P5}, constructing a solution to \eqref{EQ11} is challenging as the generator $H$ in \eqref{H} has quadratic growth in $Z$ and $\hat Z$, and exponential growth in $D$. We will tackle this in two steps. First, we construct a sequence of approximating generators $\{H^n\}_{n\in\N}$, each of which has only linear growth in $D$, such that a solution $(D^n, Z^n, \hat Z^n)$ exists by the standard results of quadratic BSDEs. Next, we show that the sequence $\{D^n\}_{n\in\N}$ is uniformly bounded from above, such that $\lim_{n\to\infty} (D^n, Z^n, \hat Z^n)$ is well-defined and actually solves \eqref{EQ11}.    

\begin{proof}[Proof of Proposition~\ref{P5}]
	For simplicity, throughout the proof we will write $\cZ_t = (Z_{t},\hat{Z}_{t})$, 
	\begin{equation}\label{SH}
	M_t = \left(1+\frac{(1-\gamma)}{\gamma}\rho_{t}^2\right),\quad \hat{M}_t = \left(1+\frac{(1-\gamma)}{\gamma}\hat{\rho}_{t}^2\right),\quad h_{t} = (1-\gamma)\left(r_{t}+\frac{\lambda_{t}^2}{2\gamma\sigma_{t}^2}\right).
	\end{equation}
	
	{\bf Step 1: Construct an approximating sequence of solutions $\{(D^{n},\cZ^{n})\}_{n\in\N}$ in $\cS_{\infty}\times\cM_{2}$.} For each $n\in\N$, consider the BSDE
	\begin{equation}\label{EQ12}
	D_{t}^{n} = \int_{t\wedge\tau}^{\tau}H^{n}(s,D_{s}^{n},Z_{s}^{n},\hat{Z}_{s}^{n})ds - \int_{t\wedge\tau}^{\tau}\cZ_{s}^{n}dB_{s}\quad t\ge 0,
	\end{equation}
	where the generator $H^{n}$ is defined by
	\begin{align}
	H^{n}(s,d,z,\hat{z})&=M_s\frac{z^{2}}{2}+\hat{M}_s\frac{\hat z^{2}}{2}+\frac{(1-\gamma)\lambda_{s}}{\gamma\sigma_{s}}\rho_{s}z+\frac{(1-\gamma)\lambda_{s}}{\gamma\sigma_{s}}\hat{\rho}_{s}\hat z+ \frac{(1-\gamma)}{\gamma}\rho_{s}\hat{\rho}_{s}z\hat z + h_{s} -\delta\theta\notag\\
	&\hspace{0.2in}+ \theta\frac{\delta^{\psi}}{\psi}\left(\mathbbm{1}_{\{|d|\le n\}}e^{-\frac{\psi}{\theta}d} + \mathbbm{1}_{\{|d|>n\}}\Big(-\frac{\psi}{\theta}d + \big(e^{\frac{-\psi}{\theta}n}+\frac{\psi}{\theta}n\big)\Big) \right).\label{H^n}
	\end{align}
	Comparing $H^n$ with $H$ in \eqref{H}, the exponential term $e^{-\frac{\psi}{\theta}d}$ is now replaced by 
	\begin{equation}\label{J}
	J(d): = \mathbbm{1}_{\{|d|\le n\}}e^{-\frac{\psi}{\theta}d} + \mathbbm{1}_{\{|d|>n\}}\big(-\frac{\psi}{\theta}d + \big(e^{\frac{-\psi}{\theta}n}+\frac{\psi}{\theta}n\big)\big).
	\end{equation}
	This ensures that $d$ grows exponentially on $[-n,n]$ and linearly otherwise with a strictly positive slope $\frac{-\psi}{\theta}$. Hence, $d\mapsto J(d)$ is continuous, strictly increasing, and of linear growth on $\R$. This, together with Assumption~\ref{A1}, implies that BSDE \eqref{EQ12} satisfies \cite[Definition 3.1 and Assumption A.1]{Confortola08}. Hence, by \cite[Theorem 3.3]{Confortola08}, there exists a unique solution $(D^{n},\cZ^{n})\in\cS_{\infty}\times\cM_{2}$ to \eqref{EQ12}. 	

	{\bf Step 2: Establish a uniform upper bound for $\{D^{n}\}_{n\in\N}$, and a solution $(D,\cZ)$ to \eqref{EQ11}.} We will construct a generator $\overline H$ such that $H^n\le \overline H$ for all $n\in\N$. With $\gamma>1$ and $\rho_{t},\hat{\rho}_t\in[-1,1]$, $M$ and $\hat M$ in \eqref{SH} satisfies $M_t, \hat M_t\in \big[\frac{1}{\gamma},1\big]$. 
	Also,	by the fact that $ab\le \frac{a^2}{2}+\frac{b^2}{2}$ for all $a,b\in\R$,  
	\[
	\bigg|\frac{(1-\gamma)\lambda_{t}}{\gamma\sigma_{t}}\eta_t z\bigg|\le \frac{(1-\gamma)^{2}\lambda_{t}^{2}}{2\gamma^{2}\sigma_{t}^{2}} + \frac{z^{2}}{2}\quad \text{and}\quad \bigg|\frac{(1-\gamma)}{\gamma}\rho_{t}\hat{\rho}_{t}z\hat{z}\bigg|\le \frac{(\gamma-1)}{\gamma}\left(\frac{z^2}{2}+\frac{\hat{z}^2}{2}\right),
	\]
	for $\eta_t =\rho_t,\hat{\rho}_t$. As a result,
	\begin{align}
	H^{n}(t,d,z,\hat{z})&\le \frac{3(z^{2}+\hat{z}^2)}{2} + \left(\frac{(1-\gamma)^{2}\lambda_{t}^{2}}{\gamma^{2}\sigma_{t}^{2}} + (1-\gamma)\left(r_{t} + \frac{\lambda_{t}^{2}}{2\gamma\sigma_{t}^{2}}\right) -\delta\theta\right)+ \theta\frac{\delta^{\psi}}{\psi}J(d)\notag\\
	&= \frac{3(z^{2}+\hat{z}^2)}{2} + (1-\gamma)\left(r_{t} + \frac{(2-\gamma)\lambda_{t}^{2}}{2\gamma^{2}\sigma_{t}^{2}}-\frac{\delta}{1-\frac{1}{\psi}}\right)+ \theta\frac{\delta^{\psi}}{\psi}J(d)\notag\\
	& \le \frac{3(z^{2}+\hat{z}^2)}{2} + C_1+ \theta\frac{\delta^{\psi}}{\psi}J(d),\label{H^n bound}
	\end{align}
	where $C_{1} := (1-\gamma)(\underline{r}-\frac{\delta}{1-\frac{1}{\psi}})$ if $\gamma\in(1,2]$, and $C_{1} := (1-\gamma)(\underline{r}-\frac{\delta}{1-\frac{1}{\psi}}) + \frac{(1-\gamma)(2-\gamma)}{2\gamma^{2}}C_{{\lambda}/{\sigma}}$ if $\gamma>2$; recall the constants $\underline{r}$, $\overline{r}$, and $C_{{\lambda}/{\sigma}}$ defined in \eqref{Market-Constants}. Now, define
	\begin{equation}\label{bar H}
	\overline{H}(d,z,\hat{z}) := \frac{3(z^{2}+\hat{z}^2)}{2}-\delta^{\psi}d+ C_{1}
	\end{equation}
	and consider the BSDE 
	\begin{equation}\label{BSDE bar}
	\overline{D} = \int_{t\wedge\tau}^{\tau}\overline{H}(\overline{D}_s,\overline{\cZ}_s)ds - \int_{t\wedge\tau}^{\tau}\overline{\cZ}_s dB_{s},\quad t\ge 0.
	\end{equation}
	Observe from \eqref{J} that
	$
	J(d)
	\ge -\frac{\psi}{\theta}d + 1\ge -\frac{\psi}{\theta}d
	$, for all $d\in\R$. 
	With $\theta<0$, this implies $\theta\frac{\delta^{\psi}}{\psi} J(d) \le -\delta^{\psi}d$ for all $d\in\R$. This, together with \eqref{H^n bound}, gives
	\begin{equation}\label{H^n<}
	H^{n}(s,d,z,\hat z)\le \overline{H}(d,z, \hat z)\ \hbox{on}\ [0,\infty)\times \R^3,\quad \hbox{for all}\ n\in\N. 
	\end{equation}
	
	Note that the generator $\overline{H}$ satisfies \cite[Definition 3.1 and Assumption A.1]{Confortola08}; particularly, it is {\it strictly} monotone in $d$. Thus, we can apply \cite[Theorem 3.3]{Confortola08} to get a unique solution $(\overline{D},\overline{\cZ})$ to \eqref{BSDE bar} in $\cS_{\infty}\times\cM_2$. Moreover, the linear dependence of $\overline H$ in $d$, along with the negative slope $-\delta^{\psi}$, indicates that \cite[Theorem 2.3]{Kobylanski00} can also be applied here (as $\overline H$ satisfies condition (ii) therein). Hence, as $(\overline{D},\overline{\cZ})$ is the {\it unique} solution to \eqref{BSDE bar} in $\cS_{\infty}\times\cM_2$, it is trivially the ``maximal solution'' in \cite[Theorem 2.3]{Kobylanski00} for which a comparison result readily holds.\footnote{There are two distinct comparison results in \cite{Kobylanski00}, i.e. Theorems 2.3 and 2.6 therein. Particularly, Theorem 2.3 allows for random horizons (so that we can apply it here), while Theorem 2.6 requires a fixed terminal time.} In view of  \eqref{H^n<}, this implies 
	\begin{equation}
	D_{t}^{n}\le \overline{D}_{t}\le \widetilde C\quad \forall t\ge 0\quad  \hbox{a.s.},\quad \text{for any}\ n\in\N,
	\end{equation}
	where $\widetilde C:=\esssup\big(\sup_{t\ge 0} \overline D_t \big)<\infty$. 
	Now, for any $n>\widetilde C$, since $D_{t}^{n}\le\widetilde C$ for all $t\ge 0$ a.s., we observe from \eqref{J} that $J(D^n_t) =  e^{-\frac{\psi}{\theta}D^n_{t}}$ for all $t\ge 0$ a.s. In view of \eqref{H^n} and \eqref{H}, we have 
	\begin{align*}
	H^{n}(t,D^n_t,Z^n_t,\hat{Z}^n_t)&= H(t,D^n_t,Z^n_t,\hat{Z}^n_t)\quad \forall t\ge 0\quad \hbox{a.s.}
	\end{align*}
	That is, $(D^n,\cZ^n)$ satisfies \eqref{EQ11}, for all $n> \widetilde C$. Specifically, $(D^n,\cZ^n) = (D^m,\cZ^m)$ for all $n,m>\widetilde C$, and $(D,\cZ) := (D^n,\cZ^n)$, for $n>\widetilde C$, is a solution to \eqref{EQ11} in $S_\infty\times \cM_2$. 
\end{proof}		

\begin{remark}\label{rem:new truncation}
	In Step 1 of the proof above, we cannot control $e^{-({\psi}/{\theta})d}$ in \eqref{H} by the truncation $e^{-({\psi}/{\theta})d}\wedge n$ used in \cite{Xing}. Indeed, to apply \cite[Theorem 3.3]{Confortola08}, the generator needs to be {\it strictly} monotone in $d$; see \cite[Assumption A.1 (ii)]{Confortola08}. Since $d\mapsto e^{-({\psi}/{\theta})d}\wedge n$ is monotone, but not strictly, 	the new truncation $J(d)$ in \eqref{J} comes into play, ensuring not only linear growth but strict monotonicity. This challenge is not present in \cite{Xing}: on a fixed horizon $T>0$ (or, $\tau\le T$ a.s.), one can apply the stronger existence result \cite[Theorem 2.3]{Kobylanski00}, which only requires the generator to be monotone, but not strictly, in $d$.
\end{remark}

\begin{remark}\label{rem:in order}
	For quadratic BSDEs on a random horizon $\tau$, \cite[Theorem 3.3]{Confortola08} gives the existence and uniqueness of solutions, while \cite[Theorem 2.3]{Kobylanski00} gives only the existence, along with comparison results for the maximal and minimal solutions. Note that when $\tau$ is unbounded, \cite[Theorem 2.3]{Kobylanski00} requires the generator to be asymptotically linear in $d$; see \cite[(H1)]{Kobylanski00}. 
	
	Consequently, in Step 2 of the proof above, we need to use \cite[Theorem 3.3]{Confortola08} first to get a unique solution. As a unique solution is trivially the maximal one, the comparison result in \cite[Theorem 2.3]{Kobylanski00} can then be invoked. This shows the importance of $\overline H$: its linear dependence on $d$, much simpler than that of $H^n$, allows us to access the comparison result in \cite[Theorem 2.3]{Kobylanski00}.  
\end{remark}

\subsection{Derivation of Lemma \ref{C4}}\label{subsec:proof for Lemma C4}
We will write $X^* = X^{\pi^*,c^*}$ for the candidate optimal wealth process, with $(\pi^*,c^*)$ defined in \eqref{optimalstrategies}. 
To begin, we investigate the integrability of $X^{*}$. Recall $C_{{\lambda}/{\sigma}}$, $\overline r$, and $\underline r$, defined in \eqref{Market-Constants}. 

\begin{lemma}\label{C3}
	Suppose $\gamma,\psi>1$ and Assumption \ref{A1} holds. Let $(D,Z,\hat{Z})\in \cS_{\infty}\times\cM_{2}$ be a solution to \eqref{EQ11}, and $(\pi^*,c^*)$ be as in \eqref{optimalstrategies}. Given $x>0$, $X^*_t>0$ for all $t\ge 0$ a.s. Moreover, for $p\ge 0$, 
	\begin{align*}
	\E\left[(X_{t}^{*})^{p}\ind_{\{t\le \tau\}}\right]&\le x^{p}\E\left[\exp\left(\left(2p\overline{r}+\left({p}+\frac{4p^2}{\gamma^2}\right)C_{{\lambda}/{\sigma}}\right)t+\frac{4p^2}{\gamma^2}\int_{0}^{t}\tZ_{s}^{2}ds\right)\ind_{\{t\le \tau\}}\right]^{{1}/{2}},\ \forall t\ge 0;
	\end{align*}
	for $p<0$, with $\widetilde{C}:=\esssup(\sup_{t\ge 0}D_s)<\infty$,
	\begin{align*}
	\E\left[(X_{\pi}^{*})^{p}\right]&\le x^p \E\left[\exp\left(\left(2p\underline{r}-2p\delta^{\psi}e^{-\frac{\psi}{\theta}\widetilde{C}}+\left(|p|+\frac{4p^2}{\gamma^2}\right)C_{{\lambda}/{\sigma}}\right)\tau+\frac{4p^2-2p}{\gamma^2}\int_{0}^{\tau}\tZ_{s}^2 ds\right)\right], \forall \pi\in\T.
	\end{align*}
\end{lemma} 

\begin{proof}
	In view of \eqref{EQ13} and \eqref{optimalstrategies}, $X^{*}$ satisfies
	\begin{align*}
	dX_{t}^{*}&= X_{t}\bigg[(r_{t}-\delta^{\psi}e^{-\frac{\psi}{\theta}D_{t}}+\frac{\lambda_{t}+\sigma_{t} \tZ_{t}}{\gamma\sigma_{t}^{2}}\lambda_{t})dt + \frac{\lambda_{t}+\sigma_{t} \widetilde Z_{t}}{\gamma\sigma_{t}}dW_{t}^{\rho}\bigg],\quad 0\le t\le \tau,
	\end{align*}
	where $\widetilde Z_{t} := \rho_{t}Z_{t}+\hat{\rho}_{t}\hat{Z}_{t}$. It follows from Assumption \ref{A1} and $(D,Z,\hat{Z})\in \cS_{\infty}\times\cM_{2}$ that $X^*_t>0$ for all $t\ge 0$ a.s. Moreover, for any $p\in\R$, 
	\begin{equation}\label{X^*}
	\left(X_{t}^{*}\right)^{p} = x^{p} \exp\left(p\int_{0}^{t\wedge\tau}(r_{t}-\delta^{\psi}e^{-\frac{\psi}{\theta}D_{t}}+ a_{s} + \frac{1}{2}(p-1)b_{s}^{2})ds\right)\cE_{t\wedge\tau}\left(L\right),
	\end{equation}
	with $a_{t} := \frac{\lambda_{t}+\sigma_{t}\tZ_{t}}{\gamma\sigma_{t}^{2}}\lambda_{t}$ , $b_{t} := \frac{\lambda_{t}+\sigma_{t}\tZ_{t}}{\gamma\sigma_{t}}$, $L_t := p\int_{0}^{t}b_{s}dW_{s}^{\rho}$, and $\cE_{t}(A)$ denoting the stochastic exponential of some process $A$.
	First, we look for a bound for $\exp(p\int_{0}^{t\wedge\tau}(a_{s} - \frac{1}{2}b_{s}^{2})ds)$. If $p\ge 0$, 
	\begin{align}\label{p>0}
	p\left(a_{t} - \frac{1}{2}b_{t}^{2}\right) & = p\left(\frac{\lambda_{t}^{2}}{\gamma^{2}\sigma_{t}^{2}}(\gamma-\frac{1}{2}) + \frac{\lambda_{t}}{\gamma^{2}\sigma_{t}}(\gamma-1)\tZ_{t} - \frac{1}{2\gamma^{2}}\tZ_{t}^{2}\right)\notag\\
	&\le p\left(\frac{\lambda_{t}^{2}}{\gamma^{2}\sigma_{t}^{2}}(\gamma-\frac{1}{2}) + \frac{\lambda_{t}^{2}}{2\gamma^{2}\sigma_{t}^{2}}(\gamma-1)^{2}\right)= p\left(\frac{\lambda_{t}^{2}}{2\sigma_{t}^{2}}\right) \le \frac{p}{2}C_{{\lambda}/{\sigma}},
	\end{align}
	where the first line follows from the definitions of $a_t$ and $b_t$, and the first inequality is due to $\frac{\lambda_{t}}{\gamma^{2}\sigma_{t}}(\gamma-1)z - \frac{1}{2\gamma^{2}}z^{2} = - \frac{1}{2\gamma^{2}}\big(z-\frac{\lambda_t}{\sigma_t}(\gamma-1)\big)^2 + \frac{\lambda_{t}^{2}}{2\gamma^{2}\sigma_{t}^{2}}(\gamma-1)^{2}$, $\forall z\in\R$. Similarly, if $p<0$, 
	\begin{align}\label{p<0}
	p\left(a_{t} - \frac{1}{2}b_{t}^{2}\right) & = p\left(\frac{\lambda_{t}^{2}}{\gamma^{2}\sigma_{t}^{2}}(\gamma-\frac{1}{2}) +\frac{\lambda_{t}}{\gamma^{2}\sigma_{t}}(\gamma-1)\tZ_{t}+ \frac{1}{2\gamma^{2}}\tZ_{t}^{2} - \frac{1}{\gamma^{2}}\tZ_{t}^{2}\right)\notag\\
	&\le p\left(- \frac{\lambda_{t}^{2}}{2\gamma^{2}\sigma_{t}^{2}}(\gamma-1)^{2} - \frac{1}{\gamma^2}\tZ_{t}^2\right)\le -\frac{p}{2} C_{\lambda/\sigma} - \frac{p}{\gamma^2} \tZ_t^2, 
	\end{align}	
	where the second line is due to $\frac{\lambda_{t}^{2}}{\gamma^{2}\sigma_{t}^{2}}(\gamma-\frac{1}{2})\ge 0$ and $\frac{\lambda_{t}}{\gamma^{2}\sigma_{t}}(\gamma-1)z + \frac{1}{2\gamma^{2}}z^{2} = \frac{1}{2\gamma^{2}}\big(z+\frac{\lambda_t}{\sigma_t}(\gamma-1)\big)^2 - \frac{\lambda_{t}^{2}}{2\gamma^{2}\sigma_{t}^{2}}(\gamma-1)^{2}$, $\forall z\in\R$. 
	Next, we look for a bound for the quadratic variation of $L$. Observe that
	\begin{align}\label{b-bound}
	\langle L\rangle_{t} &=p^{2}\int_{0}^{t\wedge\tau}|b_{s}|^{2}ds = p^{2}\int_{0}^{t\wedge\tau}\left(\frac{\lambda_{s}^{2}}{\gamma^{2}\sigma_{s}^{2}}+\frac{2\lambda_{s}}{\gamma^{2}\sigma_{s}}\tZ_{s}+\frac{1}{\gamma^{2}}\tZ_{s}^{2}\right)ds\notag \\
	&\le 2p^{2}\int_{0}^{t\wedge\tau}\left(\frac{\lambda_{s}^{2}}{\gamma^{2}\sigma_{s}^{2}}+\frac{1}{\gamma^{2}}\tZ_{s}^{2}\right)ds\le \frac{2p^2}{\gamma^2}\left(C_{{\lambda}/{\sigma}}(t\wedge\tau)+\int_{0}^{t\wedge\tau}\tZ_{s}^{2}ds\right),
	\end{align}
	where the first inequality follows from the fact that $ab \le \frac{a^2}{2}+\frac{b^2}{2}$ for any $a, b\in \R$ (by taking $a= \frac{\sqrt{2}\lambda_{s}}{\gamma\sigma_{s}}$ and $b=\frac{\sqrt{2}}{\gamma}\tZ_{s}$).
	
	Now, let us take $p\ge 0$. For any $t\ge 0$, thanks to \eqref{X^*} and $\delta^{\psi}e^{-\frac{\psi}{\theta}D_{t}}>0$,
	\begin{align*}
	\E\left[(X_{t}^{*})^{p} \ind_{\{t\le \tau\}}\right]	&\le \E\left[x^{p} \exp\left(p\int_{0}^{t\wedge\tau}\big(r_{t}+a_{s} + \frac{1}{2}(p-1)b_{s}^{2}\big)ds\right)\cE\left(L\right)_{t\wedge\tau} \ind_{\{t\le \tau\}} \right]\\
	&\le x^{p} \E\left[\text{exp}\left(\Big(p\overline{r}+\frac{p}{2}C_{{\lambda}/{\sigma}}\Big) t+\frac{1}{2}\langle L\rangle_{t}\right)\cE\left(L\right)_{t} \ind_{\{t\le \tau\}}\right],
	\end{align*}
	where	the second inequality follows from \eqref{p>0}. By direct calculation,   
	$\cE(L)_t = \cE(2L)_t^{\frac{1}{2}}\exp\left(\frac{1}{2}\langle L\rangle_t\right)$, for all $t\ge 0$.
	It then follows that
	\begin{align*}
	\E\left[(X_{t}^{*})^{p}\ind_{\{t\le \tau\}} \right]&\le x^{p}\E\left[\text{exp}\left(\Big(p\overline{r}+\frac{p}{2}C_{{\lambda}/{\sigma}}\Big)t+\frac{1}{2}\langle L\rangle_{t}\right)\cE\left(2L\right)_{t}^{\frac{1}{2}}\text{exp}\left(\frac{1}{2}\langle L\rangle_{t}\right)\ind_{\{t\le \tau\}}\right] \\
	&\le x^{p}\E\left[\text{exp}\left(2\Big(p\overline{r}+\frac{p}{2}C_{{\lambda}/{\sigma}}\Big)t+2\langle L\rangle_{t}\right) \ind_{\{t\le \tau\}}\right]^{\frac{1}{2}}\E\left[\cE\left(2 L\right)_{t}\right]^{\frac{1}{2}}\\
	&\le x^{p}\E\left[\text{exp}\left(2\Big(p\overline{r}+\frac{p}{2}C_{{\lambda}/{\sigma}}\Big)t+2\langle L\rangle_{t}\right) \ind_{\{t\le \tau\}}\right]^{\frac{1}{2}},
	\end{align*}
	where the second inequality results from applying H\"{o}lder's inequality, and the third follows from $\E[\cE\left(2L\right)_{t}]\le 1$, as $\cE\left(2L\right)$ is by definition a nonnegative local martingale, and thus a supermartingale.
	Finally, applying \eqref{b-bound} to the above inequality gives the desired result for $p\ge 0$. 
	
	For $p<0$, by using the same arguments as above for the ``$p\ge 0$'' case, except the term $\delta^{\psi}e^{-\frac{\psi}{\theta}D_{t}}$ cannot be dropped, and $1$, $\underline r$, and \eqref{p<0} replace $\ind_{\{t\le \tau\}}$, $\overline{r}$, and \eqref{p>0}, respectively, we get the desired result.  
\end{proof}

Now, we are ready to show the permissibility of $(\pi^*,c^*)$ in \eqref{optimalstrategies}.

\begin{proof}[Proof of Lemma \ref{C4}]	
	Thanks to Lemma \ref{C3} and Assumption~\ref{A4}, $\E[(X^*_\pi)^{p_-}]<\infty$ for all $\pi\in\T$. With $p_{-} = 2(2-\frac{1}{\theta})(1-\gamma)$ and $\theta<0$, this readily implies that $\{\E[(X^*_\pi)^{1-\gamma}]\}_{\pi\in\T}$ is uniformly integrable, i.e. $(X^*)^{1-\gamma}$ is of class $D$. 
	
	It remains to show that $c^{*}\in \mathcal{C}$, which, in view of \eqref{p's}, is equivalent to, 
	\begin{equation}\label{EQ15}
	\mathbb{E}\bigg[\int_{0}^{\tau}e^{-2\delta s}({c}_{s}^{*})^{p_{+}}ds\bigg]<\infty\quad \text{and}\quad \mathbb{E}\left[e^{- p_{-}\frac{\delta\theta}{1-\gamma}\tau}({c}_{\tau}^{*})^{p_{-}}\right]<\infty
	\end{equation}
	By the definitions of $c^*$ and $p_+$ in \eqref{optimalstrategies} and \eqref{p's}, 
	\begin{align*}
	({c}_{s}^{*})^{p_{+}} = (\delta^{\psi}e^{-\frac{\psi}{\theta}D_{s}}X^{*}_s)^{p_{+}} &= \delta^{2(\psi-1)}e^{\frac{-2(\psi-1)}{\theta}D_{s}}(X^{*}_s)^{p_{+}}\le \delta^{2(\psi-1)}e^{\frac{-2(\psi-1)}{\theta}\widetilde C}(X^{*}_s)^{p_{+}},
	\end{align*}
	where $\widetilde C :=\esssup (\sup_{t\ge 0}D_t)<\infty$ and the inequality is due to $\delta>0$, $\psi >1$, and $\theta<0$. Hence,
	\begin{align}\label{c permit 1}
	\mathbb{E}\bigg[\int_{0}^{\tau}e^{-2\delta s}({c}_{s}^{*})^{p_{+}}ds\bigg] &\le  \delta^{(2(\psi-1))}e^{\frac{-2(\psi-1)}{\theta}\widetilde{C}}\mathbb{E}\bigg[\int_{0}^{\tau}e^{-2\delta s}(X^{*}_s)^{p_{+}}ds\bigg].
	\end{align}
Using Fubini's theorem, we get
	\begin{align}
	&\mathbb{E}\bigg[\int_{0}^{\tau}e^{-2\delta s}(X^{*}_s)^{p_{+}}ds\bigg]=\mathbb{E}\left[\int_{0}^{\infty}e^{-2\delta s}(X^{*}_s)^{p_{+}}\mathbbm{1}_{\{s\le\tau\}}ds\right] = \int_{0}^{\infty}e^{-2\delta s}\mathbb{E}\left[(X^{*}_s)^{p_{+}}\mathbbm{1}_{\{s\le\tau\}}\right]ds\notag\\
	&\le \int_{0}^{\infty}e^{-2\delta s}\E\left[\exp\left(\left(2\overline{r}p_{+}+\left(p_{+}+\frac{4p_{+}^2}{\gamma^2}\right)C_{{\lambda}/{\sigma}}\right)\tau+\frac{4 p_{+}^2}{\gamma^2}\int_{0}^{\tau}\tZ_{u}^2 du\right)\mathbbm{1}_{\{s\le\tau\}}\right]^{1/2}ds\notag\\
	&\le \frac{1}{2\delta}\bigg(\int_{0}^{\infty}2\delta e^{-2\delta s}\E\bigg[\exp\bigg(\bigg(2\overline{r}p_{+}+\left(p_{+}+\frac{4p_{+}^2}{\gamma^2}\right)C_{{\lambda}/{\sigma}}\bigg)\tau+\frac{4 p_{+}^2}{\gamma^2}\int_{0}^\tau \tZ_{u}^2 du\bigg)\notag\\
	&\hspace{4.7in}\cdot \mathbbm{1}_{\{s\le\tau\}}\bigg]ds\bigg)^{1/2}\notag\\
	&\le \frac{1}{\sqrt{2\delta}} \E\left[\exp\left(\left(2\overline{r}p_{+}+\left(p_{+}+\frac{4p_{+}^2}{\gamma^2}\right)C_{{\lambda}/{\sigma}}\right)\tau+\frac{4 p_{+}^2}{\gamma^2}\int_{0}^\tau \tZ_{u}^2 du\right)\cdot \tau \right]^{1/2}, \label{c permit 1'}
	\end{align}
	where the first, second, and third inequalities follow from Lemma \ref{C3}, Jensen's inequality, and Fubini's theorem, respectively. Now, by taking $q>1$ specified in Assumption \ref{A4} and applying H\"{o}lder's inequality,  
	\begin{align*}
	&\E\left[\text{exp}\left(\left(2\overline{r}p_{+}+\left(p_{+}+\frac{4p_{+}^2}{\gamma^2}\right)C_{\frac{\lambda}{\sigma}}\right)\tau+\frac{4p_{+}^2}{\gamma^2}\int_{0}^{\tau}\tZ_{s}^2 ds\right)\cdot \tau\right]\\
	&\le \E\left[\text{exp}\left(q\left(2\overline{r}p_{+}+\left(p_{+}+\frac{4p_{+}^2}{\gamma^2}\right)C_{\frac{\lambda}{\sigma}}\right)\tau+\frac{4qp_{+}^2}{\gamma^2}\int_{0}^{\tau}\tZ_{s}^2 ds\right)\right]^{1/q}  \E\left[\tau^{\frac{q}{q-1}}\right]^{\frac{q-1}{q}} <\infty,
	\end{align*}
	where the finiteness is guaranteed by Assumption \ref{A4}. This, together with \eqref{c permit 1} and \eqref{c permit 1'}, establishes the first part of \eqref{EQ15}. On the other hand, a straightforward calculation, using the definitions of $c^*$ and $p_-$ and Lemma \ref{C3}, shows that the second part of \eqref{EQ15} holds under Assumption \ref{A4}. 
\end{proof}


\subsection{Proof of Theorem \ref{T1}}\label{subsec:proof of Theorem T1}
For any $(\pi,c)\in\cP$, define $R_{t}^{\pi,c}:= \frac{(X^{\pi,c}_{t\wedge\tau})^{1-\gamma}}{1-\gamma}e^{D_{t\wedge\tau}}$ and $F_{t}^{\pi,c} := R_{t}^{\pi,c} + \int_{0}^{t\wedge\tau}f(c_s,R^{\pi,c}_s)ds$, for $t\ge0$. 
Also, recall from Theorem~\ref{thm:EZ exists} the solution $(V^{c},Z^{c})$ to \eqref{EQ3}. 
In view of $H$ in \eqref{H} and the calculation in Section~\ref{subsec:ansatz}, $F^{\pi,c}$ is by construction a local supermartingale. By the Doob-Meyer decomposition and the martingale representation theorem, there exists an increasing processes $A^{\pi,c}$ and $Z^{\pi,c}$ such that $F_{t}^{\pi,c} = \int_{0}^{t\wedge\tau}Z_{s}^{\pi,c}dB_{s}-A^{\pi,c}_{t\wedge\tau}$, for all $t\ge 0$. We deduce from the definition of $F^{\pi,c}$ and its decomposition that
\begin{align*}
R^{\pi,c}_{t} 
&=\frac{(X^{\pi,c}_{\tau})^{1-\gamma}}{1-\gamma}e^{D_{\tau}} + \int_{t\wedge\tau}^{\tau}f(c_s,R^{\pi,c}_{s})ds -\int_{t\wedge\tau}^{\tau}Z_{s}^{\pi,c}dB_{s} + (A^{\pi,c}_{\tau}-A^{\pi,c}_{t\wedge\tau}),\quad t\ge 0.
\end{align*}
Noting that $D_{\tau}=0$ from \eqref{EQ11}, this shows that $(R^{\pi,c}_{t},Z^{\pi,c}_{t})_{t\ge 0}$ is a supersolution to \eqref{EQ3} with $R^{\pi,c}_\tau = V^c_\tau$.  
Then, a comparison result implies $R^{\pi,c}_{0}\ge V_{0}^{c}$.
To see this, consider $(Y_t, Z_t) := e^{-\delta\theta t}(1-\gamma) (V^c_t, Z^c_t)$ and $(\tilde Y_t, \tilde Z_t) := e^{-\delta\theta t}(1-\gamma) (R^{\pi,c}_t,Z^{\pi,c}_{t})$ for $t\ge 0$. In view of Theorem~\ref{thm:EZ exists}, $(Y, Z)$ is a solution to \eqref{EQ5} and $Y$ is of class $D$ (and thus of class $DL$). By direct calculation, $(\tilde Y, \tilde Z)$ is a subsolution to \eqref{EQ5}; moreover, as $R^{\pi,c}$ is of class $D$ (thanks to $(\pi,c)\in\mathcal P$ and $D\in\mathcal S_\infty$), $\tilde Y$ is by definition of class $DL$. Following \cite[page 246]{Xing}, we define 
\[
\alpha_t := 
\begin{cases}
\frac{F(t,c_t,Y_t)-F(t,c_t, \tilde Y_t)}{Y_t- \tilde Y_t},\quad&\hbox{for}\ Y_t\neq \tilde Y_t\\
0,\quad&\hbox{for}\ Y_t= \tilde Y_t.
\end{cases}
\]
As argued in \cite{Xing}, since $y\mapsto F(t,c,y)$ is decreasing, we have $\alpha\le 0$. By It\^{o}'s formula, we directly see that $M_\cdot:= e^{\int_0^\cdot \alpha_s ds }(Y_\cdot- \tilde Y_\cdot)$ is a local supermartingale. 
With $e^{\int_0^\cdot \alpha_s ds}$ bounded (as $\alpha\le 0$) and $Y$, $\tilde Y$ of class $DL$, $M$ is also of class $DL$ and thus a true supermartingale. 
Moreover, $M_\infty:= \lim_{t\to\infty} M_t = 0$ is well-defined, thanks to the boundedness of $e^{\int_0^\cdot \alpha_s ds}$ and $Y_t-\tilde Y_t\to Y_\tau -\tilde Y_\tau=0$ as $t\to\infty$. 
Hence, by \cite[Theorem 1.3.22]{Shreve98},  
$0=\E[M_{\tau}] \le M_{0}= Y_{0}-\tilde Y_{0}$, which implies $R^{\pi,c}_0\ge V^c_0$. 
Thus, we obtain
\[
\frac{x^{1-\gamma}}{1-\gamma}e^{D_{0}}\ge V_{0}^{c},\quad \forall(\pi,c)\in\cP. 
\]
Recall from Lemma \ref{C4} that $(\pi^{*},c^{*})\in\cP$. We will show that the upper bound $\frac{x^{1-\gamma}}{1-\gamma}e^{D_{0}}$ is achieved by $(\pi^{*},c^{*})$. Again, in view of $H$ in \eqref{H} and the calculation in Section~\ref{subsec:ansatz}, $F^{\pi^*,c^*}$ is by construction a local martingale; hence, there exists $Z^{*}$ such that $F_{t}^{\pi^*,c^*} = \int_{0}^{t\wedge\tau}Z^*_{s}dB_{s}$, for all $t\ge 0$. This gives
\[
R_{t}^{\pi^{*},c^{*}} = \frac{(X_{\tau}^{*})^{1-\gamma}}{1-\gamma}+\int_{t\wedge\tau}^{\tau}f(c_{s}^{*},R_{s}^{\pi^{*},c^{*}})ds -\int_{t\wedge\tau}^{\tau}Z_{s}^{*}dB_{s},\quad t\ge 0,
\]
implying
\[
\frac{x^{1-\gamma}}{1-\gamma}e^{D_{0}}=\E\left[\int_{0}^{\tau}f\bigg(c_{s}^{*},\frac{(X_{s}^{*})^{1-\gamma}e^{D_{s}}}{1-\gamma}\bigg)ds+ \frac{(X_{\tau}^{*})^{1-\gamma}}{1-\gamma}\right] = V^{c^*}_0,
\]
where the last equality follows from Theorem~\ref{thm:EZ exists}.



\section{The Framework of Hitting Times}\label{sec:Markov}
In this appendix, we take the random horizon $\tau$ to be the hitting time 
of a diffusion process 
to a certain region. This additional structure allows us to connect the BSDE \eqref{EQ11} to an elliptic PDE with a Dirichlet boundary condition. 
This facilitates the comparison between our random-horizon analysis and the classical one on a fixed horizon, as prior studies often use the PDE approach. We find, particularly, that the inclusion of $\hat Z$ in \eqref{optimalstrategies}-\eqref{H}, while superfluous for the fixed-horizon case, is indispensable on a random horizon; see Remark~\ref{rem:hat Z=0} for details. 
Recall that Section~\ref{subsec:Heston} employs the PDE characterization in this appendix to compute numerically the optimal strategies in the Heston model of stochastic volatility.

Let us remind ourselves of the related notation for elliptic equations in \cite{Gilbarg15}. 
Consider an open subset $\cD$ of $\R^n$, $k\in\N$, and $\nu\in(0,1)$. The H\"{o}lder space $C^{k,\nu}(\overline{\cD})$ (resp. $C^{k,\nu}(\cD)$) are defined as the subspace of $C^{k}({\cD})$ consisting of functions whose $k$th-order partial derivatives are uniformly (resp. locally) H\"{o}lder continuous with exponent $\nu$ in $\cD$.

Recall the setup in Section~\ref{subsec:setup}. In addition to $Y$ in \eqref{EQ8}, we introduce an additional state process $\cW$ given by 
\begin{equation}\label{cW}
	d\cW_{t} = \alpha(\cW_{t},Y_{t})dt + \beta(\cW_{t},Y_{t})dW_{t}+\Gamma(\cW_{t},Y_{t})d\hat{W}_{t},\quad \cW_0 = w\in \R, 
\end{equation}
for some given Borel measurable $\alpha,\beta,\Gamma:\R\times E\to\R$. 
As in \cite{Darling97, Kobylanski00, Confortola08}, we take the random horizon as the exit time of $(\cW,Y)$ from some open set $\cD\subset \R\times E$, i.e. 
\[
\tau_{w,y} :=\inf\left\{t\ge 0: (\cW_{t}^{w},Y_{t}^{y})\notin\cD \right\}.
\]
To ensure the existence of a strong solution to \eqref{cW} and sufficient regularity for subsequent analysis, 
we impose on the states $(\cW, Y)$ the following conditions, inspired by those in \cite[Section 6]{Kobylanski00}.

\begin{assump}\label{A2} $\cD\subset \R\times E$ is an open bounded set with $\partial\cD\in C^{2,\nu}$ for some $\nu\in(0,1)$. There exists an open set $U\subset \R\times E$ containing $\cD$ such that   
	\begin{itemize}
		\item [(i)] $\alpha,\beta,\Gamma$ are Lipschitz on $U$, $\inf_{U}\beta(w,y)>0$, $\inf_{U}\Gamma(w,y)>0$, and $\beta, \Gamma\in C^{1}(U)$;
		\item [(ii)] $\sigma, r, \lambda, \rho, \hat{\rho}, a$, and $b$ depend only on $y$, $\inf_{U}b(y)>0$, and $b^2\in C^{1}(U)$. %
	\end{itemize}
\end{assump}
The ellipticity conditions in Assumptions \ref{A1} and \ref{A2} guarantee the non-degeneracy of $(\cW, Y)$ in $\cD$, implying $\tau_{w,y}<\infty$ a.s. In view of \eqref{EQ8} and \eqref{cW}, the infinitesimal generator of $(\cW,Y)$ is
\[
\cL := a\frac{\partial}{\partial y} +\alpha\frac{\partial}{\partial w}+ \frac{b^{2}}{2}\frac{\partial^2}{\partial y^2} + \frac{1}{2}\left(\beta^{2}+\Gamma^2\right)\frac{\partial^2}{\partial w^{2}} + b\beta\frac{\partial^2}{\partial y\partial w},  
\]
and the corresponding elliptic boundary value problem is
\begin{equation}\label{pdeYW}
	\begin{split}
		\cL u(w,y) + G\left(y, u,\left(b\frac{\partial u}{\partial y}+\beta\frac{\partial u}{\partial w}\right),\Gamma\frac{\partial u}{\partial w}\right) &= 0,\quad (w,y)\in\cD,\\ 
		u(w,y) &=0,\quad (w,y)\in\partial\cD,
	\end{split}
\end{equation}
where 
\begin{align}\label{Hpde}
	G&(y,d,z,\hat z):=\left(1 + \frac{(1-\gamma)}{\gamma}\rho(y)^{2}\right)\frac{z^{2}}{2}+\left(1 + \frac{(1-\gamma)}{\gamma}\hat{\rho}(y)^{2}\right)\frac{{\hat z}^{2}}{2} +\frac{(1-\gamma)\lambda(y)}{\gamma\sigma(y)}\rho(y)z\notag\\
	&+\frac{(1-\gamma)\lambda(y)}{\gamma\sigma(y)}\hat{\rho}(y)\hat z + \frac{(1-\gamma)}{\gamma}\rho(y)\hat{\rho}(y)z\hat z+\frac{\delta^{\psi}\theta}{\psi}e^{-\frac{\psi}{\theta}d} + (1-\gamma)\bigg(r(y) + \frac{\lambda(y)^{2}}{2\gamma\sigma(y)^{2}}\bigg)-\delta\theta.
\end{align}

\begin{theorem}\label{PDEapproach}
	Suppose Assumptions \ref{A1} and \ref{A2} hold. Then, \eqref{pdeYW} has a unique solution $u\in C^{2,\nu}(\overline{\cD})$, with $\sup_{\cD}|\nabla u|<\infty$.
\end{theorem}

\begin{proof}
	For any $(w,y)\in\cD$, let $(D^{w,y}_{t},Z^{w,y}_{t},\hat{Z}^{w,y}_{t})\in \cS_\infty\times \cM_2$ be a solution to \eqref{EQ11} on the random horizon $\tau_{w,y}$ (obtained from Proposition~\ref{P5}), and set $\overline C := \|D^{w,y}\|_{\infty}<\infty$. By the ellipticity conditions in Assumptions \ref{A1} and \ref{A2}, $(\cW, Y)$ is non-degenerate in $\cD$, which implies that $\overline C$ is independent of the choice of $(w,y)$.
	For any $(w,y)\in{\cD}$, $d\in\R$, and $p = (p_{1},p_{2})\in\R^2$, define the function 
	\begin{equation}\label{cG} 
		\cG(w,y,d,p_1,p_2) := a(y)p_{1} + \alpha(w,y)p_{2} + G\left(y,d,b(y)p_{2}+\beta(w,y)p_{1},\Gamma(w,y)p_{1}\right),
	\end{equation}
	where $G$ is given in \eqref{Hpde}. We aim to show the existence of a solution to \eqref{PDEapproach} by using \cite[Theorem 15.12]{Gilbarg15}, which requires $\cG $ to satisfy 
	\begin{equation}\label{Gil-condition} 
		d\cdot \cG(w,y,d,0,0)\le 0\quad \text{as}\ |d|\ge M,\quad \hbox{for some}\ M>0. 
	\end{equation}
	Note that $\cG(w,y,d,0,0) =\frac{\delta^{\psi}\theta}{\psi}e^{-\frac{\psi}{\theta}d} + (1-\gamma)\big(r(y) + \frac{\lambda(y)^2}{2\gamma\sigma(y)^{2}}\big)-\delta\theta$, and thus the above condition need not hold in general. To remedy this, define $\varphi\in C^{0,1}(\R)$ by 
	\[\varphi(d):= \begin{cases} 
		d + \big(e^{-\frac{\psi}{\theta}\overline C}-\overline C\big), & d> \overline C, \\
		e^{-\frac{\psi}{\theta}d}, & -\overline C\le d\le \overline C,\\
		d + \big(e^{\frac{\psi}{\theta}\overline C}+\overline C\big), & d<-\overline C.
	\end{cases}
	\]
	Define the function $G^{\varphi}$ as $G$ in \eqref{Hpde}, with the term $e^{-\frac{\psi}{\theta}d}$ therein replaced by $\varphi(d)$. Consider the corresponding boundary value problem 
	\begin{equation}\label{PDE'}
		\begin{split}
			\cL g + G^{\varphi}(y, g,bg_{y}+\beta g_{w},\Gamma g_{w}) &=0,\text{ for }(w,y)\in\cD\\
			g(w,y)&=0,\text{ for }(w,y)\in\partial\cD.
		\end{split}
	\end{equation}
	Define $\cG^\varphi$ as in \eqref{cG} with $G$ replaced by $G^\varphi$. Note that $\cG^{\varphi}$ satisfies \eqref{Gil-condition} and is Lipschitz on $\overline{\cD}\times\R\times\R^2$ under Assumption~\ref{A1}. Hence, we can apply \cite[Theorem 15.12]{Gilbarg15} to obtain a solution $g\in C^{2,\nu}(\overline{\cD})$ to \eqref{PDE'}. Now, define the process $\overline{D}_{t}^{w,y}:= g(\cW_{t}^{w},Y_{t}^{y})$. Applying It\^{o}'s formula yields
	\begin{equation}\label{DynamicsPhi}
		d\overline{D}_{t}^{w,y} = -G^{\varphi}\left(Y_{t}^{y},\overline{D}_{t}^{w,y},\overline{Z}_{t}^{w,y},\overline{\cZ}_{t}^{w,y}\right)dt+\overline{Z}_{t}^{w,y}dW_{t} + \overline{\mathfrak Z}_{t}^{w,y}d\hat{W}_{t}
	\end{equation}
	where $\overline{Z}_{t}^{w,y} := b\frac{\partial g}{\partial y}(Y_{t}^{y},\cW_{t}^{w})+\beta \frac{\partial g}{\partial w}(Y_{t}^{y},\cW_{t}^{w})$ and $\overline{\mathfrak Z}_{t}^{w,y} := \Gamma\frac{\partial g}{\partial w}(Y_{t}^{y},\cW_{t}^{w})$. On the other hand, define $u(w,y):=D_{0}^{w,y}$ for all $(w,y)\in\overline{\cD}$. Thanks to \eqref{EQ11},
	\begin{align*}
		dD_{t}^{w,y} &= -H\left(t,D_{t}^{w,y},Z_{t}^{w,y},\hat{Z}_{t}^{w,y}\right)dt+Z_{t}^{w,y}dW_{t} + \hat{Z}_{t}^{w,y}d\hat{W}_{t}\\
		&=-G^{\varphi}\left(Y_{t}^{y},D_{t}^{w,y},Z_{t}^{w,y},\hat{Z}_{t}^{w,y}\right)dt+Z_{t}^{w,y}dW_{t} + \hat{Z}_{t}^{w,y}d\hat{W}_{t},
	\end{align*}
	where the second line follows from the definitions of $H$ and $G$ in \eqref{H} and \eqref{Hpde}, as well as $|D^{w,y}_{t}|\le \overline C$ for all $t\ge 0$. 
	Using the comparison result for quadratic BSDEs in \cite[Theorem 2.6]{Kobylanski00}, we conclude that $g(w,y) = \overline{D}_{0}^{w,y} = D_{0}^{w,y} = u(w,y)$, for all $(w,y)\in \overline{\cD}$. Hence, $u\in C^{2,\nu}(\overline{\cD})$ solves \eqref{PDE'}, and thus \eqref{pdeYW} (as $|u(w,y)|=|D_{0}^{w,y}|\le \overline C$, making $G^\varphi = G$). The uniqueness follows from the comparison principle for PDEs with quadratic growth in \cite[Theorem 1.2]{Barles95}.
	\end{proof}

Now, let $(D,Z,\hat{Z})$ be the solution to the BSDE \eqref{EQ11} on the random horizon $\tau_{w,y}$, obtained in Proposition~\ref{P5}. The connection between \eqref{EQ11} and \eqref{pdeYW} can be stated precisely through the following $\P$-a.s. representation: for all $t\ge 0$,
\begin{align}\label{representation}
	D_{t} = u(\cW_{t}^{w},Y_{t}^{y}),\quad Z_{t} = b\frac{\partial u}{\partial y}(\cW_{t}^{w},Y_{t}^{y})+\beta\frac{\partial u}{\partial w}(\cW_{t}^{w},Y_{t}^{y}),\quad \hat{Z}_{t}=\Gamma\frac{\partial u}{\partial w}(\cW_{t}^{w},Y_{t}^{y}),
\end{align}
where $u\in C^{2,\nu}(\overline{\cD})$ is the solution to \eqref{pdeYW}. This follows simply by applying It\^{o}'s formula to $u$. 
An important message of \eqref{representation} is that $\hat Z$ can be dropped completely on a fixed horizon, but is indispensable in general when a random horizon is considered. 

\begin{remark}\label{rem:hat Z=0}
	If $\Gamma(w,y)\equiv 0$ in \eqref{cW}, the randomness of $\tau_{w,y}$ comes exclusively from $W$. Then, \eqref{representation} indicates $\hat Z\equiv 0$, implying that one can drop $\hat Z$ completely in Section~\ref{subsec:ansatz}. This particularly covers the standard case with a fixed horizon $T>0$, by taking $\alpha\equiv 1$, $\beta\equiv 0$, $\Gamma\equiv 0$ in \eqref{cW} and $\cD = (w-\eps,w+T)\times E$, for any $\eps>0$. With $\hat Z\equiv 0$, the setup in Section~\ref{subsec:ansatz} is consistent with those in \cite{Xing, Pham02, Xing18, Kraft13, Kraft17} on a fixed horizon. In particular, $\hat Z\equiv 0$ leads to simpler $(\pi^*,c^*)$ and $H$ in \eqref{optimalstrategies} and \eqref{H}, which recover \cite[Theorem 5.1]{Kraft17} and \cite[(2.12), (2.13)]{Xing}. 
	
	When the randomness of $\tau_{w,y}$ comes jointly from $W$ and $\hat W$ (i.e. $\Gamma(w,y)\not\equiv 0$ in \eqref{cW}), \eqref{representation} indicates that $\hat Z$ is not identically zero and thus cannot be omitted in general.  
\end{remark}


\section{Proofs for Section \ref{Examples}}\label{section: proofs for examples}
\subsection{Proof of Propostition~\ref{Prop: Parabolic}}\label{subsec: Parabolic PDE}
	By \cite[Theorem 4.9]{Kraft17}, the Cauchy problem
	\begin{equation*}
	\begin{split}
		u_{t}(t,y)+\cH(u,-\sigma u_{y}) = \Big(\alpha-\frac{1}{2}\sigma^{2}\Big)u_{y}(t,y) - \frac{1}{2}\sigma^{2}u_{yy}(t,y),\quad &\forall y\in\R,\ t\in[0,T), \\
		u(T,y) = 1,\quad &\forall y\in\R,
	\end{split}
\end{equation*}
has a unique positive solution $h\in C^{1,2}([0,T]\times\R)$ that is bounded from above and away from zero; moreover, $\|h_y\|_\infty<\infty$. Take $0<\underline{h}<\overline{h}$ such that $0<\underline{h}\le h(t,y)\le \overline{h}$ for all $(t,y)\in[0,T)\times\R$. Define $\overline{u}(t,y) := h(t,y) + (\underline{h}+1)$ and $\hat{u}(t,y) := h(t,y) - (\overline{h}+1)$. A direct calculation confirms that $\overline{u}\text{ (resp. }\hat{u})$ is an ordered upper (resp. lower) solution of \eqref{PDE_Ex1}; see \cite[p.99, Definition 2.1]{Pao92}. Since $h$ and $h_{y}$ are bounded, the conditions of \cite[Theorem 4.2]{Pao92} are met, which gives the existence of a unique bounded solution $u\in C^{1,2}([0,T]\times\R)$ to \eqref{PDE_Ex1}.


\subsection{Derivation of Lemma~\ref{lem:MGF}}\label{subsec:proof of Lemma MGF}
To prove Lemma~\ref{lem:MGF}, we need to first characterize the survival probability of $\tau_w$.
\begin{lemma}\label{lem:density}
	For any $L>0$ and $(w,y)\in\cD$, 
	\begin{align}\label{prob}
	P(w,y,t) &:= \P(\tau_w > t \mid \cW_0=w,\ Y_0=y)\nonumber\\
	&= \sum_{n=0}^{\infty}\frac{4(-1)^{n}}{\pi(2n+1)}\exp\left(-A_{n}(\alpha t)-\frac{2\alpha}{k^{2}}B_{n}(\alpha t) y\right)\cos\left(\frac{(2n+1)\pi w}{L}\right)\quad \forall t\ge 0,
	\end{align}
	where, under the notation $ \beta_{n} := \frac{k}{\alpha}(2n+1)\pi$ and $\eps_{n}:=\left(\frac{(2n+1)\pi}{\sqrt{2\alpha}L}\right)^{2}\eps$,
	\begin{align}
	A_{n}(s) &:= \frac{2\alpha m^{2}}{k^{2}}\ln\left(\frac{(\Delta_{n}+1) + (\Delta_{n}-1)e^{-\Delta_{n}s}}{2\Delta_{n}}\right)+\left(\frac{\alpha m^2(\Delta_{n}-1)}{k^2}+\eps_{n}\right)s,\label{A_n}\\
	B_{n}(s) &:=\frac{\beta_{n}^{2}}{2L^{2}}\left[\frac{1-e^{-\Delta_{n}s}}{(\Delta_{n}+1)+ (\Delta_{n}-1)e^{-\Delta_{n}s}}\right],\quad \hbox{with}\quad \Delta_{n} := \sqrt{1+({\beta_{n}}/{L})^{2}}.\label{B_n}
	\end{align}
\end{lemma}

\begin{proof}
The associated backward Fokker-Planck equation of $P(w,y,t)$ is
\begin{equation}\label{HM-FokkerPlanck}
\frac{\partial P}{\partial t} = -\alpha(y-m^{2})\frac{\partial P}{\partial y} + \frac{1}{2}k^{2}y\frac{\partial^{2}P}{\partial y^{2}} + \frac{1}{2}(y+\eps)\frac{\partial^{2}P}{\partial w^{2}},
\end{equation}
with initial condition $P(w,y,0)=1$ and boundary condition $P(\pm \frac{L}{2},y,t)=0$. Similarly to \cite{Masoliver08}, as $w\mapsto P(w,y,t)$ is an even function (due to \eqref{mean return}), it can be expressed as a Fourier series, i.e.
\[
P(w,y,t) = \sum_{n=0}^{\infty}P_{n}(y,t)\cos\left(\frac{(2n+1)\pi w}{L}\right)
\]
with the Fourier coefficients
\begin{equation}\label{P_n}
P_{n}(y,t) = \frac{2}{L}\int_{-L/2}^{L/2}P(w,y,t)\cos\left(\frac{(2n+1)\pi w}{L}\right)dw. 
\end{equation}
In view of \eqref{HM-FokkerPlanck} and \eqref{P_n}, the change of variables $s = \alpha t$ and $v=(\frac{2\alpha}{k^{2}})y$ gives the equation  
\begin{equation}\label{HM-FokkerPlanck-Fourier}
\frac{\partial P_{n}}{\partial s} = -(v-\mu)\frac{\partial P_{n}}{\partial v} + v\frac{\partial^{2}P_{n}}{\partial v^{2}} - \left(\frac{\beta_{n}}{2L}\right)^{2}vP_{n} + \eps_{n}P_{n},
\end{equation}
with initial condition $P_{n}(v,0) = \frac{4(-1)^{n}}{\pi(2n+1)}$, where $\mu :=  \frac{2\alpha m^{2}}{k^{2}}$ and $\beta_n$, $\eps_n$ are defined in the statement of Lemma~\ref{lem:density}. This can be solved by the ansatz
\[
P_{n}(v,s) = \frac{4(-1)^{n}}{\pi(2n+1)}\text{exp}\left(-A_{n}(s)-B_{n}(s)v\right).
\]
Differentiating and substituting this back into \eqref{HM-FokkerPlanck-Fourier}, we find
\begin{equation*}
A_{n}'(s) = -v\left(B_{n}'(s)+B_{n}(s)+B_{n}^{2}(s)-\left(\frac{\beta_{n}}{2L}\right)^{2}\right)+\mu B_{n}(s)+\eps_{n}.
\end{equation*}
Notice that $B_{n}(s)$ must solve the Riccati equation 
\begin{equation}\label{Riccatti}
B_{n}'(s) = -B_{n}(s)-B_{n}^{2}(s) + \left(\frac{\beta_{n}}{2L}\right)^{2},\quad B_{n}(0) = 0,
\end{equation}
under which
$A_{n}'(s) = \mu B_{n}(s) + \eps_{n}$, implying $A_{n}(s)= \mu\int_{0}^{s}B_{n}(t)dt + \eps_{n}s$. The solution to \eqref{Riccatti}, derived in \cite{Masoliver08}, is given as in \eqref{B_n}. 
It follows that one can calculate $A_n(s)$ as in \eqref{A_n}.
Therefore, the survival probability has the representation  
\begin{equation}\label{Survival}
P(w,v,s) = \sum_{n=0}^{\infty}\frac{4(-1)^{n}}{\pi(2n+1)}\text{exp}\left(-A_{n}(s)-B_{n}(s)v\right)\cos\left(\frac{(2n+1)\pi w}{L}\right), 
\end{equation}
with $A_n, B_n$ specified as above. Changing the variables $(v,s)$ back to $(y,t)$ gives \eqref{prob}.
\end{proof}

Now, we are ready to prove Lemma~\ref{lem:MGF}. We will use the same notation in \eqref{A_n} and \eqref{B_n}. 

\begin{proof}[Proof of Lemma~\ref{lem:MGF}]
Recall $P(w,y,t)$ in \eqref{prob} and define $F^{w,y}(t):=1-P(w,y,t)$, $t\in[0,\infty)$, the distribution function of $\tau_w$. For any $c\in [0,c^*)$, thanks to integration by parts, 
\begin{align}\label{to show MGF}
\E\left[e^{c\tau_w}\right] &= \int_{0}^{\infty}e^{ct}d F^{w,y}(t)= -\lim_{t\to\infty} e^{ct} P(w,y,t) +1+ c \int_{0}^{\infty}e^{ct}P(w,y,t)dt.
\end{align}
To show that $\E[e^{c\tau_w}]<\infty$, it suffices to prove $\int_{0}^{\infty}e^{ct}P(w,y,t)dt<\infty$ (because it readily implies $\lim_{t\to\infty}e^{ct} P(w,y,t)=0$). 
As $c\in [0,c^*)$, take $\eta\in[0,1)$ such that $c=\eta c^*$. Observe from \eqref{A_n} that  
\begin{align*}
A_{n}(\alpha t) \ge -\frac{2\alpha m^2}{k^2}\ln(2) + \left( \left(\frac{\alpha m}{k}\right)^2 (\Delta_{n}-1)+(2n+1)^2 \frac{\pi^2\eps}{2L^2}\right) t,\quad \forall t\ge 0. 
\end{align*}
Hence, by the nonnegativity of $B_n(\alpha t)$ in \eqref{B_n}, 
\begin{equation}\label{ct-A-B}
ct-A_{n}(\alpha t)-\frac{2\alpha}{k^{2}}B_{n}(\alpha t)y \le \eta c^* t- A_{n}(\alpha t)\le \frac{2\alpha m^2}{k^2}\ln(2) -\left(4 n^2+4n+(1-\eta)\right)\frac{\pi^2\eps}{2L^2}t, 
\end{equation}
where the second inequality follows from the definition of $c^*$ in \eqref{c^*}. Now, by \eqref{prob} and \eqref{ct-A-B},
\begin{align}\label{e^ctP}
\int_{0}^{\infty}e^{ct}P(w,y,t)dt &\le \frac{4}{\pi} \int_0^\infty \sum_{n=0}^{\infty} \exp\left(ct -A_{n}(\alpha t)-\frac{2\alpha}{k^{2}}B_{n}(\alpha t) y\right) dt\notag\\
&\le \frac{2^{2 (\frac{\alpha m^2}{k^2}+1)} }{\pi} \int_0^\infty \sum_{n=0}^{\infty} \exp\left(-\left(4 n^2+4n+(1-\eta)\right)\frac{\pi^2\eps}{2L^2}t\right) dt\notag\\
&\le \frac{2^{2 (\frac{\alpha m^2}{k^2}+1)} }{\pi} \int_0^\infty e^{-\left(1-\eta\right)\frac{\pi^2\eps}{2L^2}t} \bigg(1+ \sum_{n=1}^{\infty} e^{-\frac{2\pi^2\eps t}{L^2}n^2}\bigg) dt.
\end{align}
Noting that
\[
\sum_{n=1}^{\infty} e^{-\frac{2\pi^2\eps t}{L^2}n^2} \le \int_0^\infty e^{-\frac{2\pi^2\eps t}{L^2}x^2} dx = \frac{L}{2^{3/2}\sqrt{\pi\eps t}} ,
\]
we conclude from \eqref{e^ctP} that
\begin{align*}
\int_{0}^{\infty}e^{ct}P(w,y,t)dt 
&\le \frac{2^{2 (\frac{\alpha m^2}{k^2}+1)} }{\pi} \int_0^\infty e^{-\left(1-\eta\right)\frac{\pi^2\eps}{2L^2}t} \bigg(1+ \frac{L  t^{-1/2}}{2^{3/2}\sqrt{\pi\eps}}\bigg) dt<\infty,
\end{align*}
as desired.
\end{proof}

\bibliographystyle{siam}
\bibliography{refsMathSci}

\begin{thebibliography}{10}

\bibitem{AH20}
{\sc J.~Aurand and Y.-J. Huang}, {\em Mortality and healthcare: a stochastic
  control analysis under {E}pstein-{Z}in preferences}, SIAM J. Control Optim.,
  59 (2021), pp.~4051--4080.

\bibitem{Bansal07}
{\sc R.~Bansal}, {\em Long-run risks and financial markets}, Review - Federal
  Reserve Bank of St.Louis, 89 (2007), pp.~283--299.

\bibitem{Bansal04}
{\sc R.~Bansal and A.~Yaron}, {\em Risks for the long run: A potential
  resolution of asset pricing puzzles}, The Journal of Finance, 59 (2004),
  pp.~1481--1509.

\bibitem{Barberis98}
{\sc N.~Barberis, A.~Shleifer, and R.~Vishny}, {\em A model of investor
  sentiment}, Journal of Financial Economics, 49 (1998), pp.~307--343.

\bibitem{Barles95}
{\sc G.~Barles and F.~Murat}, {\em Uniqueness and the maximum principle for
  quasilinear elliptic equations with quadratic growth conditions}, Arch.
  Rational Mech. Anal., 133 (1995), pp.~77--101.

\bibitem{Benzoni11}
{\sc L.~Benzoni, P.~Collin-Dufresne, and R.~S. Goldstein}, {\em Explaining
  asset pricing puzzles associated with the 1987 market crash}, Journal of
  Financial Economics, 101 (2011), pp.~552--573.

\bibitem{Bhamra10}
{\sc H.~S. Bhamra, L.-A. Kuehn, and I.~A. Strebulaev}, {\em The levered equity
  risk premium and credit spreads: A unified framework}, The Review of
  Financial Studies, 23 (2010), pp.~645--703.

\bibitem{BC76}
{\sc F.~Black and J.~C. Cox}, {\em Valuing corporate securities: Some effects
  of bond indenture provisions}, Journal of Finance, 31 (1976), pp.~351--67.

\bibitem{Blanchet08}
{\sc C.~Blanchet-Scalliet, N.~El~Karoui, M.~Jeanblanc, and L.~Martellini}, {\em
  Optimal investment decisions when time-horizon is uncertain}, J. Math.
  Econom., 44 (2008), pp.~1100--1113.

\bibitem{Bouchard04}
{\sc B.~Bouchard and H.~Pham}, {\em Wealth-path dependent utility maximization
  in incomplete markets}, Finance Stoch., 8 (2004), pp.~579--603.

\bibitem{Briand00}
{\sc P.~Briand and R.~Carmona}, {\em B{SDE}s with polynomial growth
  generators}, J. Appl. Math. Stochastic Anal., 13 (2000), pp.~207--238.

\bibitem{Confortola08}
{\sc P.~Briand and F.~Confortola}, {\em Quadratic {BSDE}s with random terminal
  time and elliptic {PDE}s in infinite dimension}, Electron. J. Probab., 13
  (2008), pp.~1529--1561.

\bibitem{Briand03}
{\sc P.~Briand, B.~Delyon, Y.~Hu, E.~Pardoux, and L.~Stoica}, {\em Lp solutions
  of backward stochastic differential equations}, Stochastic Processes and
  their Applications, 108 (2003), pp.~109 -- 129.

\bibitem{Briand06}
{\sc P.~Briand and Y.~Hu}, {\em B{SDE} with quadratic growth and unbounded
  terminal value}, Probab. Theory Related Fields, 136 (2006), pp.~604--618.

\bibitem{Briand07}
\leavevmode\vrule height 2pt depth -1.6pt width 23pt, {\em Quadratic {BSDE}s
  with convex generators and unbounded terminal conditions}, Probab. Theory
  Related Fields, 141 (2008), pp.~543--567.

\bibitem{Cheridito11}
{\sc P.~Cheridito and Y.~Hu}, {\em Optimal consumption and investment in
  incomplete markets with general constraints}, Stoch. Dyn., 11 (2011),
  pp.~283--299.

\bibitem{Chernov00}
{\sc M.~Chernov and E.~Ghysels}, {\em A study towards a unified approach to the
  joint estimation of objective and risk neutral measures for the purpose of
  options valuation}, Journal of Financial Economics, 56 (2000), pp.~407--458.

\bibitem{Darling97}
{\sc R.~W.~R. Darling and E.~Pardoux}, {\em Backwards sde with random terminal
  time and applications to semilinear elliptic pde}, Ann. Probab., 25 (1997),
  pp.~1135--1159.

\bibitem{DeLong90}
{\sc J.~B. De~Long, A.~Shleifer, L.~Summers, and R.~Waldmann}, {\em Noise
  trader risk in financial markets}, Journal of Political Economy, 98 (1990),
  pp.~703--738.

\bibitem{Dragulescu02}
{\sc A.~D. Dr\u{a}gulescu and V.~M. Yakovenko}, {\em Probability distribution
  of returns in the {H}eston model with stochastic volatility}, Quant. Finance,
  2 (2002), pp.~443--453.

\bibitem{DuffieEpstein92}
{\sc D.~Duffie and L.~G. Epstein}, {\em Stochastic differential utility},
  Econometrica, 60 (1992), pp.~353--394.
\newblock With an appendix by the authors and C. Skiadas.

\bibitem{DuffieLions}
{\sc D.~Duffie and P.-L. Lions}, {\em P{DE} solutions of stochastic
  differential utility}, J. Math. Econom., 21 (1992), pp.~577--606.

\bibitem{ElKaroui10}
{\sc N.~El~Karoui, M.~Jeanblanc, and Y.~Jiao}, {\em What happens after a
  default: the conditional density approach}, Stochastic Process. Appl., 120
  (2010), pp.~1011--1032.

\bibitem{EZ89}
{\sc L.~Epstein and S.~Zin}, {\em Substitution, risk aversion, and the temporal
  behavior of consumption and asset returns: A theoretical framework},
  Econometrica, 57 (1989), pp.~937--69.

\bibitem{Fama65}
{\sc E.~F. Fama}, {\em The behavior of stock-market prices}, The Journal of
  Business, 38 (1965), pp.~34--105.

\bibitem{Friedman53}
{\sc M.~Friedman}, {\em The case for flexible exchange rates}, in Essays in
  Positive Economics, University of Chicago Press, 1953.

\bibitem{Gilbarg15}
{\sc D.~Gilbarg and N.~S. Trudinger}, {\em Elliptic partial differential
  equations of second order}, Classics in Mathematics, Springer-Verlag, Berlin,
  2001.
\newblock Reprint of the 1998 edition.

\bibitem{GH19}
{\sc P.~Guasoni and Y.-J. Huang}, {\em Consumption, investment and healthcare
  with aging}, Finance Stoch., 23 (2019), pp.~313--358.

\bibitem{Hansen07}
{\sc L.~P. Hansen, J.~Heaton, J.~Lee, and N.~Roussanov}, {\em Intertemporal
  substitution and risk aversion}, in Handbook of Econometrics, J.~J. Heckman
  and E.~E. Leamer, eds., vol.~6A, Elsevier, 1~ed., 2007, ch.~61,
  pp.~3967--4056.

\bibitem{Hu05}
{\sc Y.~Hu, P.~Imkeller, and M.~M\"{u}ller}, {\em Utility maximization in
  incomplete markets}, Ann. Appl. Probab., 15 (2005), pp.~1691--1712.

\bibitem{Janecek12}
{\sc K.~Jane\v{c}ek and M.~S\^{i}rbu}, {\em Optimal investment with
  high-watermark performance fee}, SIAM J. Control Optim., 50 (2012),
  pp.~790--819.

\bibitem{Monique15}
{\sc M.~Jeanblanc, T.~Mastrolia, D.~Possama\"{i}, and A.~R\'{e}veillac}, {\em
  Utility maximization with random horizon: a {BSDE} approach}, Int. J. Theor.
  Appl. Finance, 18 (2015), p.~1550045.

\bibitem{Shreve98}
{\sc I.~Karatzas and S.~E. Shreve}, {\em Brownian motion and stochastic
  calculus}, vol.~113 of Graduate Texts in Mathematics, Springer-Verlag, New
  York, second~ed., 1991.

\bibitem{Karatzas00}
{\sc I.~Karatzas and H.~Wang}, {\em Utility maximization with discretionary
  stopping}, SIAM J. Control Optim., 39 (2000), pp.~306--329.

\bibitem{Kharroubi13}
{\sc I.~Kharroubi, T.~Lim, and A.~Ngoupeyou}, {\em Mean-variance hedging on
  uncertain time horizon in a market with a jump}, Appl. Math. Optim., 68
  (2013), pp.~413--444.

\bibitem{Kobylanski00}
{\sc M.~Kobylanski}, {\em Backward stochastic differential equations and
  partial differential equations with quadratic growth}, Ann. Probab., 28
  (2000), pp.~558--602.

\bibitem{Korn03}
{\sc R.~KORN and H.~KRAFT}, {\em Optimal portfolios with defaultable securities
  a firm value approach}, International Journal of Theoretical and Applied
  Finance, 06 (2003), pp.~793--819.

\bibitem{Kraft17}
{\sc H.~Kraft, T.~Seiferling, and F.~T. Seifried}, {\em Optimal consumption and
  investment with {E}pstein-{Z}in recursive utility}, Finance Stoch., 21
  (2017), pp.~187--226.

\bibitem{Kraft13}
{\sc H.~Kraft, F.~T. Seifried, and M.~Steffensen}, {\em Consumption-portfolio
  optimization with recursive utility in incomplete markets}, Finance Stoch.,
  17 (2013), pp.~161--196.

\bibitem{Kraft06}
{\sc H.~Kraft and M.~Steffensen}, {\em Portfolio problems stopping at first
  hitting time with application to default risk}, Math. Methods Oper. Res., 63
  (2006), pp.~123--150.

\bibitem{KP78}
{\sc D.~Kreps and E.~L. Porteus}, {\em Temporal resolution of uncertainty and
  dynamic choice theory}, Econometrica, 46 (1978), pp.~185--200.

\bibitem{Masoliver08}
{\sc J.~Masoliver and J.~Perell\'{o}}, {\em Escape problem under stochastic
  volatility: the {H}eston model}, Phys. Rev. E (3), 78 (2008), p.~056104.

\bibitem{Masoliver09}
\leavevmode\vrule height 2pt depth -1.6pt width 23pt, {\em First-passage and
  risk evaluation under stochastic volatility}, Phys. Rev. E (3), 80 (2009),
  p.~016108.

\bibitem{Xing18}
{\sc A.~Matoussi and H.~Xing}, {\em Convex duality for {E}pstein-{Z}in
  stochastic differential utility}, Math. Finance, 28 (2018), pp.~991--1019.

\bibitem{Melnyk20}
{\sc Y.~Melnyk, J.~Muhle-Karbe, and F.~T. Seifried}, {\em Lifetime investment
  and consumption with recursive preferences and small transaction costs},
  Mathematical Finance, 30 (2020), pp.~1135--1167.

\bibitem{Merton69}
{\sc R.~C. Merton}, {\em Lifetime portfolio selection under uncertainty: The
  continuous-time case}, The Review of Economics and Statistics, 51 (1969),
  pp.~247--257.

\bibitem{Morlais09}
{\sc M.-A. Morlais}, {\em Quadratic {BSDE}s driven by a continuous martingale
  and applications to the utility maximization problem}, Finance Stoch., 13
  (2009), pp.~121--150.

\bibitem{Pan02}
{\sc J.~Pan}, {\em The jump-risk premia implicit in options: evidence from an
  integrated time-series study}, Journal of Financial Economics, 63 (2002),
  pp.~3--50.

\bibitem{Pao92}
{\sc C.~V. Pao}, {\em Nonlinear parabolic and elliptic equations}, Plenum
  Press, New York, 1992.

\bibitem{Pardoux99}
{\sc E.~Pardoux}, {\em B{SDE}s, weak convergence and homogenization of
  semilinear {PDE}s}, in Nonlinear analysis, differential equations and control
  ({M}ontreal, {QC}, 1998), vol.~528 of NATO Sci. Ser. C Math. Phys. Sci.,
  Kluwer Acad. Publ., Dordrecht, 1999, pp.~503--549.

\bibitem{Pham02}
{\sc H.~Pham}, {\em Smooth solutions to optimal investment models with
  stochastic volatilities and portfolio constraints}, Appl. Math. Optim., 46
  (2002), pp.~55--78.

\bibitem{Pham-book-09}
\leavevmode\vrule height 2pt depth -1.6pt width 23pt, {\em Continuous-time
  stochastic control and optimization with financial applications}, vol.~61 of
  Stochastic Modelling and Applied Probability, Springer-Verlag, Berlin, 2009.

\bibitem{Royer04}
{\sc M.~Royer}, {\em Bsdes with a random terminal time driven by a monotone
  generator and their links with pdes}, Stochastics and Stochastic Reports, 76
  (2004), pp.~281--307.

\bibitem{Schroder96}
{\sc M.~Schroder and C.~Skiadas}, {\em Optimal consumption and portfolio
  selection with stochastic differential utility}, J. Econom. Theory, 89
  (1999), pp.~68--126.

\bibitem{Seifried16}
{\sc T.~Seiferling and F.~T. Seifried}, {\em Epstein-zin stochastic
  differential utility: Existence, uniqueness, concavity, and utility
  gradients},  (2016).
\newblock Preprint. Available at https://dx.doi.org/10.2139/ssrn.2625800.

\bibitem{SV97}
{\sc A.~Shleifer and R.~Vishny}, {\em The limits of arbitrage}, Journal of
  Finance, 52 (1997), pp.~35--55.

\bibitem{Silva03}
{\sc A.~C. Silva and V.~M. Yakovenko}, {\em Comparison between the probability
  distribution of returns in the {H}eston model and empirical data for stock
  indexes}, Phys. A, 324 (2003), pp.~303--310.
\newblock International Econophysics Conference IEC2002 (Bali).

\bibitem{Skiadas98}
{\sc C.~Skiadas}, {\em Recursive utility and preferences for information},
  Economic Theory, 12 (1998), pp.~293--312.

\bibitem{Vissing03}
{\sc A.~Vissing-J{\o}rgensen and O.~Attanasio}, {\em Stock-market
  participation, intertemporal substitution, and risk-aversion}, American
  Economic Review, 93 (2003), pp.~383--391.

\bibitem{Xing}
{\sc H.~Xing}, {\em Consumption-investment optimization with {E}pstein-{Z}in
  utility in incomplete markets}, Finance Stoch., 21 (2017), pp.~227--262.

\bibitem{Yaari65}
{\sc M.~E. Yaari}, {\em Uncertain lifetime, life insurance, and the theory of
  the consumer}, The Review of Economic Studies, 32 (1965), pp.~137--150.

\bibitem{Zariphopoulou01}
{\sc T.~Zariphopoulou}, {\em A solution approach to valuation with unhedgeable
  risks}, Finance Stoch., 5 (2001), pp.~61--82.

\bibitem{Zawisza15}
{\sc D.~Zawisza}, {\em Robust consumption-investment problem on infinite
  horizon}, Appl. Math. Optim., 72 (2015), pp.~469--491.

\end{thebibliography}
\end{document}